\documentclass[runningheads]{llncs}

\input{macro}

\begin{document}
    \title{Fine-grained Causal Reversibility for Asynchronous Channel-based Programming%
    \thanks{An extended version of the paper accepted at ICTAC 2026. This version includes complete proofs and additional technical material.}}

    \author{
        Shunya Oguchi\inst{1} \and
        Shoji Yuen\inst{1}\orcidID{0000-0003-2642-0647} \and\\
        Nobuko Yoshida\inst{2}\orcidID{0000-0002-3925-8557} \and\\
        Claudio Antares Mezzina\inst{3}\orcidID{0000-0003-1556-2623}
    }

    \authorrunning{Shunya Oguchi et al.}
    \titlerunning{Causal Reversibility for Asynchronous Channel-based Programming}

    \institute{Graduate School of Informatics, Nagoya University, Japan\\
    \email{\{oguchi321,yuen\}@sqlab.jp}
    \and
    Department of Computer Science, University of Oxford, England, United Kingdom
    \email{nobuko.yoshida@cs.ox.ac.uk}\\
    \and
    University of Bari Aldo Moro, Italy\\
    \email{claudio.mezzina@uniba.it}
    }

    \maketitle

    \begin{abstract}
        Causal reversibility has emerged as an effective technique for debugging concurrent systems.
        In particular, rolling back and replaying a concurrent program with causal consistency has been found very helpful in debugging concurrency bugs.
        In channel-based communication, queue ordering creates dependencies that prevent
        causally independent actions from being rolled back and replayed.
        To enable efficient rollback and replay without being constrained by queue dependencies, it is necessary to analyse causal dependencies between messages in queues and reorder independent messages.
        This paper presents revGo, a core of the Go programming language
        \sy{assuming unbounded asynchronous channels.}
        Our rollback-and-replay semantics allows us to \sy{reorder} messages in the queue if they are not causally related in the forward execution.
        \sy{It is shown that reordering independent messages generates no configuration with non-reachable processes.
        By reordering independent messages,} rollback and replay are implementable with minimality by assigning unique keys to communications.

    \end{abstract}


    \section{Introduction}
    \label{sec:introduction}
    
    Programming distributed systems is inherently challenging due to the presence of concurrency, partial failures, and unpredictable communication delays.
    Even when individual components are simple, their interactions can lead to complex global behaviours that are difficult to reason about and even harder to debug.
    Modern programming languages increasingly provide concurrency and communication primitives to help developers structure distributed programs,
    but these abstractions do not eliminate the fundamental difficulty of understanding and diagnosing erroneous behaviours that arise from concurrent interactions.

    Such erroneous behaviours are known as {\em concurrency bugs}~\cite{conc_bugs}.
    Concurrency bugs are particularly problematic because they are typically non-deterministic and hard to reproduce: rerunning the program may not trigger the same faulty behaviour, and small changes in timing or load may hide the bug altogether.
    Undoing an action in a concurrent or distributed setting should not arbitrarily erase unrelated progress; rather, it should affect only those actions that are causally dependent on it.
    This observation is at the basis of \emph{causally consistent reversible debugging}~\cite{GiachinoLM14,LaneseG24},
    which preserves the causal structure of concurrent computations during rollback and replay.
    In the setting of actor-based languages, several reversible debuggers have been developed for message-passing systems, notably for Erlang~\cite{actorverse,cauder}.
    Among these, CauDEr~\cite{cauder} supports causal-consistent rollback and replay of concurrent executions in an actor system.
    CauDEr exploits Erlang's asynchronous message-passing semantics and mailbox-based communication to track causal dependencies between send and receive actions,
    enabling selective undo while preserving causal correctness.
    \sy{While communications in Erlang are by the mailbox, Go~\cite{donovan2015go} processes communicate by channel queues.
    Channel queues have the dependencies due to message ordering in channels, which do not exist in the Erlang semantics.
    To handle this dependency of message queues, rollback and replay require safely reordering independent messages in channel queues.}

    This work introduces revGo, a \og{calculus} with a rollback-and-replay semantics \og{that assumes unbounded asynchronous channels and abstracts away from Go's bounded and synchronized channels},
    aiming at supporting effective reversible debugging for concurrent and distributed Go programs.
    Unlike message-passing communication in actor-based models,
    \emph{channel-based communication}~\cite{MiGo} introduces additional causal structures due to channel queues.

    By enriching Go's operational semantics with histories for processes and channels,
    revGo enables not only reversible executions but also selective rollback-and-replay of executions,
    allowing developers to focus on the actions relevant to a failure.
    In revGo, dependencies among messages in channel queues must be considered when reversing executions.
    In rollback-and-replay semantics in particular, messages in channel queues that are irrelevant to
    a failure have to be abstracted away.

    The following example illustrates the potential uncertainty of channel-based communications
    when a channel queue is shared among multiple threads.
    The behaviour is highly dependent on the runtime, such as goroutine scheduling and channel queue capacities.
    \begin{example}
        \label{ex:multi}
        \quad
        \begin{myfigure}[-2.5em][-2.3em]
            \begin{center}
                
    \begin{minipage}{0.9\textwidth}
        \begin{lstlisting}[basicstyle=\ttfamily\small,columns=fullflexible,escapeinside={(*@}{@*)}]
 func S(c chan int, n int) { c <- n; S(c, n) }
 func R(c chan int) { n := <-c; if (*@\truered{2/n}@*) == 2 { R(c) } }
 func main() { c := make(chan int, *); go S(c, 0); go S(c, 1); R(c) }
        \end{lstlisting}
    \end{minipage}

            \end{center}
            \caption{A program example where a division-by-zero error may arise.}
            \label{fig:running}
        \end{myfigure}
    \end{example}

    In the example, traditional debugging techniques cannot easily isolate the faulty interaction, since replaying the program may lead to a different message ordering,
    and stepping backwards would undo many actions that are causally independent of the failure.
    The challenge is therefore to identify and undo only the concurrent actions that caused \texttt{main} to receive $0$.

    Our main contribution is to introduce a causally-consistent equivalence over channel queues.
    The semantics identifies channel queues via equivalence classes, e.g. all possible message reorderings.
    To ensure that message reorderings are sound with respect to reachability,
    we attach a unique \emph{key}~\cite{CCSK,RevMiGo} to each message to track independence in behaviour.
    As long as communication actions involving the keys are independent, the messages are independent.
    In the above example, all messages with 1 are independent of those with 0 since the two processes \texttt{S} are independent.
    Due to independence, even if many messages with 1 are communicated before or after the message with 0 that caused the error, all the messages with 1 are irrelevant to the error.

    \sy{It is shown that reordering independent messages does not generate unreachable processes, thereby establishing the soundness of the reordering technique.
    This property enables greater flexibility in managing channel queues during reversible debugging to reach a target event.
    Under the reordering of independent messages,
    the rollback-and-replay semantics exhibits both soundness and minimality in reaching a rollback action, as observed in Erlang~\cite{LamiLSCF24}.
    }

    \paragraph{Related Work:}
    The idea of exploiting causal-consistent reversible semantics for debugging purposes was first introduced in~\cite{GiachinoLM14} in the context of
    a subset of the Oz programming language~\cite{RoyH2004}.
    This approach was subsequently developed and refined for Erlang, leading to the implementation of the CauDEr reversible debugger~\cite{cauder,revErl,LamiLSCF24,LanesePV21},
    which supports causal-consistent rollback and replay of concurrent executions in actor-based systems.
    We depart from these works in several respects.
    Although the language considered in~\cite{GiachinoLM14} features channel-based communication, its reversible semantics does not support queue reordering,
    and thus cannot abstract away from causally independent message orderings.
    In contrast, CauDEr targets Erlang, an actor-based language with mailbox semantics,
    where message ordering does not play the same causal role as in channel-based systems.
    Our approach instead focuses on explicit FIFO channels, where message ordering is an essential part of the causal structure,
    and introduces a formal mechanism to safely reorder causally independent messages.
    \og{\cite{DBLP:journals/lmcs/MezzinaP21} studies reversible semantics for asynchronous multiparty session types with higher-order messages over a single global queue.}

    \sy{
        RevMiGo\cite{RevMiGo} proposed by the authors gives a reversible semantics by constructing a graph structure to
        store the dependency among communications. The graph is globally shared by all processes, whereas this paper proposes a semantics different from RevMiGo, where the dependency is stored in each process and each channel queue
        in a distributed manner.
    }

    \paragraph{Paper Organisation}

     Section \ref{sec:revgo} introduces revGo as a core language with channel-based communication and presents
    its reversible operational semantics with causal consistency.
    Section \ref{sec:rr-semantics} presents a rollback-and-replay semantics,
    and Section \ref{sec:rr-properties} discusses key reversible properties of this semantics and its usefulness for reversible debugging,
    based on soundness and minimality.
    Section \ref{sec:conclusion} concludes.

    \section{revGo}
    \label{sec:revgo}
    
    This section presents revGo.
    We use revGo as a minimal core of Go to investigate causal-consistent reversibility for channel-based communication via queues.
    Unlike Go, processes communicate via channels with unbounded queues.

    \subsection{Syntax}
    
    Fig.~\ref{fig:syntax} presents the syntax of \emph{revGo},
    where $v$ ranges over \emph{values}, $x$ over \emph{variables}, $e$ over \emph{expressions}, $X$ over \emph{process variables}, \og{$w$ over \emph{channel variables}, and $c$ over \emph{channel names}}.
    We write $\seq{v}$ for a sequence $v_1\cdots v_n (n \geq 0)$.
    We similarly use $\seq{x}$\og{, $\seq{w}$, $\seq{c}$, $\seq{u}$,} and $\seq{e}$.
    A \emph{process term} $P$ has the following forms:
    (1) \emph{channel creation} $\nwcS{\og{w}}{\sigma}{P_1 \| \cdots \| P_n}$ binds $w$ and
    generates a fresh channel of type $\sigma$ and spawns $n$ subprocesses in parallel.
    (2) \emph{send} $\sndS{\og{u}}{e}[P']$ sends the value of $e$ on \og{$u$}.
    (3) \emph{receive} $\rcvS{\og{u}}{x}[P']$ receives a value from \og{$u$} and binds it to $x$.
    (4) \emph{if} $\ifS{e}{P_1}{P_2}$ branches according to $e$.
    (5) \emph{select} $\slcS{\pi_i:P_i}_{i \in I}$ $(|I| \geq 2)$ non-deterministically selects an available channel operation.
    (6) \emph{invocation} $\defS{X}{\seq{e},\og{\seq{u}}}$ calls a process variable $X$ with parameters.
    (7) $\tauS{P'}$ performs a silent action.
    (8) $\zroS$ is \emph{an inactive process term}.
    \og{Prefixes $\sndPi{c}{e}$ and $\rcvPi{c}{x}$ are runtime forms obtained when $\mathtt{newchan}$ substitutes $c$ for $w$ in $\sndPi{w}{e}$ and $\rcvPi{w}{x}$.
    When an invocation is executed, it takes the form $X\langle\seq{e},\seq{c}\rangle$.}
    Process definitions are collected in a program $\program \bnfdef \{D_i\}_{i\in I}\,\mathtt{in}\,P$,
    where each definition has the form $X(\seq{x}\og{,\seq{w}})=P$.

    \ogm{We use $\cdot$ for concatenation, sometimes omit $\zroS$, and write $\fv{P}$ for the free variables of $P$.}

    \begin{figure}[H]
        \input{fig0_syntax}
        \caption{Syntax of revGo}
        \label{fig:syntax}
    \end{figure}

\end{document}

    \subsection{Operational Semantics}
    
    \sy{
        The operational semantics of revGo is characterized by two components: processes and channel queues.
        Both components maintain a history to reverse actions. A process records the sequence
        of communications with channel queues during forward execution and performs rollbacks based on the
        recorded history. A channel queue maintains the order of messages that have been enqueued
        as well as the messages that have been dequeued during forward execution. During backward execution, the channel queue
        communicates with processes in the reverse order established during forward execution.
    }
    \sy{
        Each channel queue operates as a first-in, first-out (FIFO) structure of messages, each associated with a unique key. The key of a message identifies the asynchronous communication
        among processes. Keys are stored within the histories of process actions.}

    \sy{First, the reversible process semantics and the behavior of channel queues in enqueueing and dequeueing message sequences are presented.
        Finally, the operational semantics are defined by composing processes and channel queues.
    }
    \label{subsec:processes}
    \begin{definition}[Process]\label{def:process}
        A process is a tuple $\pcfg{\varp}{h}{P}$, where
        $\varp$ is a pid, $h$ is a process history, and $P$ is a process term.
        \ogm{A process configuration $\Pi$ is a parallel composition of processes with $|$, which is associative and commutative.}
    \end{definition}

    Given a process $\pcfg{\varp}{h}{P}$, $\varp$ is a unique process identifier.
    We use $\rootp$ for the initial process.
    A \emph{process history} $h$ is a sequence of \emph{history items} where a history item has one of the following:
    $\nwcH{c}{\seq{\varp}}$; $\sndH{c}{k}{e}{\mathit{slc}}$; $\rcvH{c}{k}{P}{\mathit{slc}}$;
    $\tauH{\mathit{slc}}$; $\ifH{e}{P}$; and $\defH{X}{\seq{e},\seq{c}}$.
    \og{
        $\nwcH{c}{\seq{\varp}}$ stores the created $c$ and $\seq{\varp}$ of the spawned processes.
        $\sndH{c}{k}{e}{\mathit{slc}}$ and $\rcvH{c}{k}{P}{\mathit{slc}}$ store the channel $c$ and key $k$ of the communication;
        for a select, $\mathit{slc}$ is the set of unselected branches $\{\pi_i:P_i\}_{i\in I}$, and otherwise it is $\bot$; the same applies to $\tauH{\mathit{slc}}$.
        $\defH{X}{\seq{e},\seq{c}}$ stores the parameters of invocation.
        Some history items record $e$ and $P$ for reversibility.
    }
    Given a revGo program $\program = \{D_i\}_{i\in I} \,\mathtt{in}\, P$, its \emph{initial process configuration} $\Pini^{\program}$ is $\pcfg{\rootp}{\varepsilon}{P}$, where $P$ is the initial process term of $\program$.

    \og{We define \emph{forward and backward actions} as transition labels.
    A forward action $\alpha$ has one of the forms listed in Table~\ref{tab:action}.}
    We write $\extA{\varp}{i}$ for a forward external action.
    For $\alpha$, the backward action is denoted by $\overline{\alpha}$ and represents the undoing of $\alpha$.
    Let $\mathcal{A}$ be the set of forward actions and $\overline{\mathcal{A}}$ be the set of backward actions.

    \begin{table}[H]
        \centering
        \caption{Forms of forward actions}
        \label{tab:action}
        \begin{tabular}{@{}p{0.2\textwidth}@{\hspace{4pt}}p{0.4\textwidth}|p{0.38\textwidth}@{}}
            \hline
            $\itnA{\varp}{i}$                 & Performs an internal action.                & \multirow{4}{0.38\textwidth} {$\varp$: performing process's pid\\ $i$: action index in process $\varp$} \\
            $\nwcA{\varp}{i}{c}{\seq{\varp}}$ & Creates $c$ and subprocesses $\seq{\varp}$. & \\
            $\sndA{\varp}{i}{c}{k}{v}$        & Sends $(k,v)$ to $c$.                       & \\
            $\rcvA{\varp}{i}{c}{k}{v}$        & Receives $(k,v)$ from $c$.                  & \\
            \hline
        \end{tabular}
    \end{table}

    Fig.~\ref{fig:process_fwd} defines forward transitions over process configurations, written $\Pi \proB{\alpha} \Pi'$.
    We write $P\{v/x\}$ for the standard capture-avoiding substitution, where all free occurrences of $x$ in $P$ are replaced with $v$.
    \og{We write $e\downarrow v$ when $e$ is evaluated to be $v$.
    We define $\mathsf{alt}$ by $\mathsf{alt}(\sndH{c}{k}{e}{\bot},\{\pi_i:P_i\}_{i\in I})=\sndH{c}{k}{e}{\{\pi_i:P_i\}_{i\in I}}$,
        and analogously for $\textcolor{myPurple}{\mathsf{rcv}}$ and $\textcolor{myPurple}{\mathsf{tau}}$.}
    \og{\frname{Nwc} substitutes $c$ for $w$ in the spawned subprocesses.
    \frname{Snd} evaluates the expression to get the send value.
    \frname{Rcv} substitutes the received value for $x$.
    For these external actions, the channel, key, and received value are supplied by the queue.}
    \frname{Tau} \og{performs} a silent action.
    \frname{IfT} and \frname{IfF} \og{proceed with the branch determined by the guard and record the untaken branch}.
    \frname{Def} evaluates arguments and unfolds a definition from the program $\program$.
    In \frname{Slc} the unselected alternatives are recorded in the generated history item so that the choice can be reconstructed during backward execution.

    \begin{figure}[H]
        \centering
        
    \begin{tabular}{rll}
        \frname{Nwc} & $\Pi\mid\pcfg{\varp}{h}{\nwcS{\og{w}}{\sigma}{P_1\|\cdots\| P_n}}
        \proB{\nwcA{\varp}{|h|}{c}{\varp[1] \cdots \varp[n]}}$ \\
        & $\Pi\mid\pcfg{\varp}{h\cdot\nwcH{c}{\varp[1] \cdots \varp[n]}}{\zroS}\mid\miniprod_{i=1}^{n} \pcfg{\varp[i]}{\varepsilon}{\og{P_i\{c/w\}}}$ \\[0.4em]

        \frname{Snd} & $\irule{e \downarrow v}{\Pi\mid\pcfg{\varp}{h}{\sndS{c}{e}[P]}
        \proB{\sndA{\varp}{|h|}{c}{k}{v}} \Pi\mid\pcfg{\varp}{h\cdot\sndH{c}{k}{e}{\bot}}{P}}$ \\[1.4em]

        \frname{Rcv} & $\Pi\mid\pcfg{\varp}{h}{\rcvS{c}{x}[P]}
        \proB{\rcvA{\varp}{|h|}{c}{k}{v}} \Pi\mid\pcfg{\varp}{h\cdot\rcvH{c}{k}{P}{\bot}}{P\{v/x\}}$ \\[0.4em]

        \frname{Tau} & $\Pi\mid\pcfg{\varp}{h}{\tauS{P}}
        \proB{\itnA{\varp}{|h|}} \Pi\mid\pcfg{\varp}{h\cdot\tauH{\bot}}{P}$ \\[0.5em]

        \frname{IfT} & $\irule{e \downarrow \mathtt{true}}{\Pi\mid\pcfg{\varp}{h}{\ifS{e}{P_1}{P_2}}
        \proB{\itnA{\varp}{|h|}} \Pi\mid\pcfg{\varp}{h\cdot\ifH{e}{P_2}}{P_1}}$ \\[1.3em]

        \frname{IfF} & $\irule{e \downarrow \mathtt{false}}{\Pi\mid\pcfg{\varp}{h}{\ifS{e}{P_1}{P_2}}
        \proB{\itnA{\varp}{|h|}} \Pi\mid\pcfg{\varp}{h\cdot\ifH{e}{P_1}}{P_2}}$ \\[1.3em]

        \frname{Def} & $\irule{e_i \downarrow v_i \quad \og{X(\seq{x},\seq{w}) = P'} \in \program}{\Pi\mid\pcfg{\varp}{h}{\defS{X}{\seq{e},\seq{c}}}
        \proB{\itnA{\varp}{|h|}} \Pi\mid\pcfg{\varp}{h\cdot\defH{X}{\seq{e},\seq{c}}}{P'\{\seq{v},\seq{c}/\seq{x}\og{,\seq{w}}\}}}$ \\[1.3em]

        \frname{Slc} & $\irule{\Pi\mid\pcfg{\varp}{h}{\pi_j:P_j}\ \proB{\alpha}\ \Pi'\mid\pcfg{\varp}{h\cdot \varr}{P'} \quad
        \varr[\mathsf{alt}] = \mathsf{alt}(\varr,\{\pi_i:P_i\}_{i \in I\setminus\{j\}})}
        {\Pi\mid\pcfg{\varp}{h}{\slcS{\pi_i:P_i}_{i \in I}}\ \proB{\alpha}\ \Pi'\mid\pcfg{\varp}{h\cdot\varr[\mathsf{alt}]}{P'}}$ &
        \quad \\

    \end{tabular}
\end{document}

        \caption{Forward semantics for process configurations}
        \label{fig:process_fwd}
    \end{figure}

    Fig.~\ref{fig:process_bwd} defines backward transitions over process configurations, written $\Pi \proB{\overline{\alpha}} \Pi'$.
    In the backward execution, a process restores the previous configuration
    according to the information \og{in the top history item}.
    We omit \brname{Tau} and \brname{IfF}.

    \begin{figure}[H]
        \centering
        
    \begin{tabular}{rll}
        \brname{Nwc} & $\irule{\og{w \notin \fv{P_i}}}{\Pi\mid\pcfg{\varp}{h\cdot\nwcH{c}{\varp[1] \cdots \varp[n]}}{\zroS}\mid\miniprod_{i=1}^{n} \pcfg{\varp[i]}{\varepsilon}{P_i}\proB{\overline{\nwcA{\varp}{|h|}{c}{\varp[1] \cdots \varp[n]}}}}$ \\
        & $(\Pi\mid\pcfg{\varp}{h}{\nwcS{\og{w}}{\sigma}{\og{P_1\{w/c\}}\|\cdots\| \og{P_n\{w/c\}}}})$ \\[0.6em]

        \brname{Snd} & $\irule{e\downarrow v}{\Pi\mid\pcfg{\varp}{h\cdot\sndH{c}{k}{e}{\bot}}{P}\proB{\overline{\sndA{\varp}{|h|}{c}{k}{v}}}(\Pi\mid\pcfg{\varp}{h}{\sndS{c}{e}[P]})}$ \\[1.4em]

        \brname{Rcv} & $\Pi\mid\pcfg{\varp}{h\cdot\rcvH{c}{k}{P}{\bot}}{P'}\proB{\overline{\rcvA{\varp}{|h|}{c}{k}{v}}}(\Pi\mid\pcfg{\varp}{h}{\rcvS{c}{x}[P]})$ \\[0.9em]


        \brname{IfT} & $\irule{e \downarrow \mathtt{true}}{\Pi\mid\pcfg{\varp}{h\cdot\ifH{e}{P_2}}{P_1}
        \proB{\overline{\itnA{\varp}{|h|}}} \Pi\mid\pcfg{\varp}{h}{\ifS{e}{P_1}{P_2}}}$ \\[1.4em]


        \brname{Def} & $\Pi\mid\pcfg{\varp}{h\cdot\defH{X}{\seq{e},\seq{c}}}{P'} \proB{\overline{\itnA{\varp}{|h|}}} \Pi\mid\pcfg{\varp}{h}{\defS{X}{\seq{e},\seq{c}}}$ \\[0.7em]

        \brname{Slc} & $\irule{\Pi\mid\pcfg{\varp}{h}{\pi_j:P_j}\ \proB{\overline{\alpha}}\ \Pi'\mid\pcfg{\varp}{h}{P} \quad \varr[\mathsf{alt}] = \mathsf{alt}(\varr,\{\pi_i:P_i\}_{i \in I\setminus\{j\}})}
        {\Pi\mid\pcfg{\varp}{h\cdot\varr[\mathsf{alt}]}{P'}\ \proB{\overline{\alpha}}\ \Pi'\mid\pcfg{\varp}{h}{\slcS{\pi_i:P_i}_{i \in I}}}$ & \quad \\

    \end{tabular}
\end{document}

        \caption{Backward semantics for process configurations}
        \label{fig:process_bwd}
    \end{figure}

    \begin{definition}[Channel queue]
        \label{def:queue}
        A channel queue is a tuple $\qcfg{c}{\sigma}{\mu}{\mu'}$, where $c$ is a channel name,
        $\sigma$ is a type, and $\mu$ and $\mu'$ are sequences of messages.
        A channel queue configuration $Q$ \ogm{is a composition of channel queues with $|$.}
    \end{definition}

    A channel queue $\qcfg{c}{\sigma}{\mu}{\mu'}$ is a FIFO where $\sigma$ is the type of messages,
    $\mu$ is the message sequence, and $\mu'$ is the history of messages sent to the queue.
    The message immediately at the left of $\circ$ is the head of the current queue.
    When a process sends a new message, the message is appended at the left end of $\mu$,
    and when a process receives a message, the message is removed from $\mu$ and appended at the left end of $\mu'$.
    The empty queue is $\circ$.
    Since we assume no channel queue before execution, the initial queue configuration, $\Qini$, is $\varnothing$.

    Fig.~\ref{fig:queue_fwd} defines reversible transitions over channel queue configurations, written $Q \queBRev{\extA{\varp}{i}} Q'$.
    Bidirectional arrows represent both forward and backward transitions.
    In the forward direction, \rrname{QNwc} creates a channel queue with a fresh ${c}$,
    \rrname{QSnd} adds a message with a fresh key $k$ to $\mu$,
    and \rrname{QRcv} moves a message from $\mu$ to $\mu'$ as the record of communication.
    In the backward direction, \rrname{QNwc} removes a queue,
    \rrname{QSnd} removes the leftmost message in $\mu$, and \rrname{QRcv} moves the rightmost message in $\mu'$ to $\mu$.

    \begin{figure}[H]
        \centering
        
    \begin{tabular}{rll}
        \rrname{QNwc} & $Q \queBRevLong{\nwcA{\varp}{i}{c}{\seq{\varp}}} Q \mid \qcfg{c}{\sigma}{\varepsilon}{\varepsilon} (c \text{ is fresh.})$ \\[0.2em]

        \rrname{QSnd} & $Q \mid \qcfg{c}{\sigma}{\mu}{\mu'} \queBRevLong{\sndA{\varp}{i}{c}{k}{v}} Q \mid \qcfg{c}{\sigma}{(k,v)\cdot\mu}{\mu'} (k \text{ is fresh.})$ \\[0.2em]

        \rrname{QRcv} & $Q \mid \qcfg{c}{\sigma}{\mu\cdot(k,v)}{\mu'} \queBRevLong{\rcvA{\varp}{i}{c}{k}{v}} Q \mid \qcfg{c}{\sigma}{\mu}{(k,v)\cdot\mu'}$ \\[0.2em]
    \end{tabular}

        \caption{Reversible semantics for channel queue configurations}
        \label{fig:queue_fwd}
    \end{figure}

    \begin{definition}[System configuration]
        A \emph{system configuration} $\Sys$ is a pair $(\Pi,Q)$ of a process configuration and a queue configuration.
    \end{definition}

    Given a revGo program $\program$, let $\Ini^{\program} = (\Pini^{\program},\Qini)$ be the initial system configuration for $\program$.
    Fig.~\ref{fig:config} defines forward and backward transitions over system configurations,
    written $\Sys \genB{\alpha} \Sys'$ and $\Sys \genB{\overline{\alpha}} \Sys'$.
    \frname{Itn} and \brname{Itn} lift an internal transition of the process configuration to a system transition, leaving the channel queues unchanged.
    \frname{Ext} and \brname{Ext} synchronise a process transition with the corresponding channel queue transition by matching the action.

    \begin{figure}[H]
        \centering
        \input{fig5_system}
        \caption{Forward and backward semantics for system configurations}
        \label{fig:config}
    \end{figure}

    \begin{example}
        \label{ex:basic}
        Let $\program = \{S,R\}\,\mathtt{in}\,P_0$ be the revGo program from Example~\ref{ex:multi}, where
        \begin{center}
            $\begin{aligned}
                 \og{S(n,w)} &= \og{\sndS{w}{n}[\defS{S}{n,w}]} \quad
                 \og{R(w) = \rcvS{w}{n}[\ifS{2/n == 2}{\defS{R}{w}}{\zroS}]}\\
                 P_0 &= \og{\nwcS{w}{\mathtt{int}}{\defS{S}{0,w}\|\defS{S}{1,w}\|\defS{R}{w}}}
            \end{aligned}$
        \end{center}

        Starting from $\Ini^{\program} = (\pcfg{\rootp}{\varepsilon}{P_0},\varnothing)$, we show transitions:

        \begin{align*}
            \Ini^{\program} \genB{\nwcA{\rootp}{0}{c}{\seq{\varp}}}&
            (\Pi_1,\qcfg{c}{\texttt{int}}{\varepsilon}{\varepsilon}) \genBStar{\rho} (\Pi_2,\qcfg{c}{\texttt{int}}{\varepsilon}{\varepsilon})\\
            \genB{\sndA{\varp[1]}{1}{c}{k_1}{1}}&
            (\Pi_3,\qcfg{c}{\texttt{int}}{\red{(k_1,1)}}{\varepsilon})
            \!\genB{\rcvA{\varp[2]}{1}{c}{k_1}{1}}\!(\Pi_4,\qcfg{c}{\texttt{int}}{\varepsilon}{\blue{(k_1,1)}})\\
            \genB{\sndA{\varp[0]}{1}{c}{k_0}{0}}&
            (\Pi_5,\qcfg{c}{\texttt{int}}{\red{(k_0,0)}}{(k_1,1)})\\
            \genBStar{\rho'}&
            (\Pi_6,\qcfg{c}{\texttt{int}}{\red{(k_n,1)\cdots(k_2,1)}(k_0,0)}{(k_1,1)})\\
            \genB{\rcvA{\varp[2]}{4}{c}{k_0}{0}}&
            (\Pi_7,\qcfg{c}{\texttt{int}}{(k_n,1)\cdots(k_2,1)}{\blue{(k_0,0)}(k_1,1)}) = \Tar
        \end{align*}
        \og{First, we create the channel $c$ and spawn the three processes $\varp[0],\varp[1],\varp[2]$ corresponding to $\defS{S}{0,c}$, $\defS{S}{1,c}$, and $\defS{R}{c}$, respectively.}
        \frname{Def} is applied for $\rho$.
        Then $\varp[1]$ sends $1$ on $c$ and $\varp[2]$ receives it.
        Next $\varp[0]$ sends $0$ on $c$, and $\rho'$ performs $n-1$ additional sends of $1$ by $\varp[1]$ on $c$ (so there are $2$ to $n$ occurrences of $1$ in total).
        Finally $\varp[2]$ receives $0$ from $c$ immediately before the division-by-zero.

    \end{example}

\end{document}

    \subsection{Causal-consistent Reversibility}
    \label{subsec:reversibility}
    
    We now show that revGo operational semantics satisfies the key properties of a causal-consistent reversible semantics:
    causal consistency, causal safety, and causal liveness~\cite{10.1145/3648474}.
    Following the axiomatic framework of~\cite{10.1145/3648474}, we define an \emph{independence relation} over actions.

    We observe the dependency among actions derived from the control flow of processes defined as the \emph{must-happen-before} relation.
    \begin{definition}[Must-Happen-before]
        \label{def:must-happen-before}
        For $\alpha,\alpha'\in\mathcal{A}$,
        we say $\alpha$ immediately happens before $\alpha'$, written $\alpha \mhb \alpha'$, if one of the following holds:
        \begin{trivlist}
            \item[(Sequential)] $\varp = \varp[']$ and $i + 1= i'$;
            \item[(Creation)] $\alpha=\nwcA{\varp}{i}{c}{\varp[1]\cdots \varp[n]}$, $\varp[']\in \{\varp[1]\cdots \varp[n]\}$, and $i'=1$; and
            \item[(Send-receive)] $\alpha=\sndA{\varp}{i}{c}{k}{\_}$ and $\alpha'=\rcvA{\varp[']}{i'}{c}{k}{\_}$,
        \end{trivlist}
        where $\varp,\varp[']$ are the pids, and $i,i'$ are the indices, of $\alpha,\alpha'$, respectively.
    \end{definition}

    The dependency from must-happen-before does not include causality of channel queues.
    The independence relation over actions is defined as the complement of the dependency of the
    must-happen-before relation and the queue ordering as follows.
    \ogm{We use $\beta \in \mathcal{A}\cup\overline{\mathcal{A}}$ for an action.}
    \og{For $\alpha\in\mathcal{A}$, let $\und{\alpha}=\und{\overline{\alpha}}=\alpha$.}
    \begin{definition}[Independence relation]
        We define the independence relation on actions, written $\idpB$, as follows.
        For $\beta_1,\beta_2\in\mathcal{A}\cup\overline{\mathcal{A}}$, let $\alpha_1=\und{\beta_1}$ and $\alpha_2=\und{\beta_2}$.
        Then $\beta_1 \idpB \beta_2$ iff the following conditions hold:
        \begin{itemize}
            \item $\alpha_1 \centernot{\mhb} \alpha_2$ and $\alpha_2 \centernot{\mhb} \alpha_1$; and
            \item if $\alpha_1=\sndA{\_}{\_}{c}{\_}{\_}$ and $\alpha_2=\sndA{\_}{\_}{c'}{\_}{\_}$, then $c \neq c'$; and
            \item if $\alpha_1=\rcvA{\_}{\_}{c}{\_}{\_}$ and $\alpha_2=\rcvA{\_}{\_}{c'}{\_}{\_}$, then $c \neq c'$.
        \end{itemize}
    \end{definition}

    Following~\cite{10.1145/3648474}, we define \emph{causal equivalence} to identify derivations that differ only by swapping independent transitions
    and cancelling a transition with its reverse.
    \ogm{For \og{an action sequence} $\rho = \beta_1 \cdots \beta_n \in(\mathcal{A}\cup\overline{\mathcal{A}})^*$, we write $\Sys_0 \genBStar{\rho} \Sys_n$ if\\

    \noindent
        $\Sys_{i-1} \genB{\beta_i} \Sys_i$ for $1 \leq i \leq n$.
        We use $t:\Sys\genB{\beta}\Sys'$ for a transition and $\overline{t}:\Sys'\genB{\overline{\beta}}\Sys$ for its reverse.
        A \emph{derivation} is a transition sequence.
    }
    \begin{definition}[Causal equivalence]
        \label{def:causal-equivalence}
        Let $\cequivN$ be the smallest equivalence relation on derivations closed under concatenation and satisfying:
        \begin{trivlist}
            \item[(Swap)] if $\beta_1 \idpB \beta_2$ and there exist co-initial transitions
            $t_1: \Sys \genB{\beta_1} \Sys_1$ and $t_2: \Sys \genB{\beta_2} \Sys_2$,
            and co-final transitions $t'_2: \Sys_1 \genB{\beta_2} \Sys'$ and $t'_1: \Sys_2 \genB{\beta_1} \Sys'$, then
            $t_1 t'_2 \cequivN t_2 t'_1$.
            \item[(Cancellation)] for every transition $t$, $t\,\overline{t} \cequivN \varepsilon$ and $\overline{t}\,t \cequivN \varepsilon$.
        \end{trivlist}
    \end{definition}

    \ogm{We say that $\Sys$ is \emph{reachable} from $\program$ if $\Ini^{\program} \genBStar{\rho} \Sys$ for some $\rho$,
    and that $\Sys$ is \emph{reachable} if  $\Sys$ is reachable from some $\program$.}
    We establish the reversibility properties for reachable system configurations.
    A fundamental property for any reversible calculus is the loop lemma stating that every action can be undone.
    \begin{lemma}[Loop lemma]
        \label{lem:loop-system}
        \og{For reachable $\Sys$ and $\Sys'$,}
        $\Sys \genB{\alpha} \Sys'$ iff $\Sys' \genB{\overline{\alpha}} \Sys$.
    \end{lemma}
    \begin{proof}
        Let $\Sys=(\Pi,Q)$ and $\Sys'=(\Pi',Q')$.
        It is immediate that $\Pi' \proB{\overline{\alpha}} \Pi$ if $\Pi \proB{\alpha} \Pi'$
        since the history item contains the information to reconstruct the previous process configuration.
        If $\alpha$ is internal, then the system rules leave the queue component unchanged.
        If $\alpha$ is external, the channel-queue relation is defined bidirectionally, so
        $Q \queBRev{\alpha} Q'$ holds exactly when the corresponding reverse queue step is available.
        The other direction relies on reachability since
        \brname{Rcv} and \brname{Def} can be executed under weaker conditions than \frname{Rcv} and \frname{Def}, respectively.
        From the backward rules, we can see that the result of a backward transition is uniquely determined:
        if $\Pi' \proB{\overline{\alpha}} \Pi_1$ and $\Pi' \proB{\overline{\alpha}} \Pi_2$, then $\Pi_1 = \Pi_2$.
        Using this with the fact that $\Pi' \proB{\overline{\alpha}} \Pi$ if $\Pi \proB{\alpha} \Pi'$,
        we can conclude that if $\Pi'$ is reachable and $\Pi' \proB{\overline{\alpha}} \Pi$, then $\Pi \proB{\alpha} \Pi'$ holds.
        Combining the process transition with the bidirectional queue transition in the system rules gives
        $\Sys \genB{\alpha} \Sys'$ iff $\Sys' \genB{\overline{\alpha}} \Sys$.
    \end{proof}

    The transition relation is shown to satisfy the following properties:
    Square property (SP), Backward transitions are independent (BTI), and Well-foundedness (WF).

    \begin{lemma}[Square property]
        \label{lem:square-property}
        Given $\program$, suppose that $\Ini^{\program} \genBStar{} \Sys$, $\Sys \genB{\beta_1} \Sys_1$, $\Sys \genB{\beta_2} \Sys_2$, and $\beta_1 \idpB \beta_2$.
        Then, there is $\Sys'$ such that $\Sys_1 \genB{\beta_2} \Sys'$ and $\Sys_2 \genB{\beta_1} \Sys'$.
    \end{lemma}
    \begin{proof}
        The cases where conflicts must be considered are the following:
        actions of the same process;
        actions of a parent process and its child process;
        and send or receive actions on the same channel.

        First, the condition $\beta_1 \idpB \beta_2$ excludes actions of the same process.

        Second, consider the case of actions of a parent process and its child process.
        It suffices to consider the case where the parent process action corresponds to channel creation.
        If the parent process action is a forward transition, then the child process cannot perform either a forward or a backward action, leading to a contradiction.
        If the parent process action is a backward transition, then the child process undoes its actions.
        However, under $\beta_1 \idpB \beta_2$, the child process cannot execute its first action, which is again a contradiction.

        Finally, consider the case where both actions operate on the same channel.
        By the condition $\beta_1 \idpB \beta_2$, one action must be a send action and the other a receive action, and their keys must be distinct.
        A send action either adds or removes a message at the left end of the contents, whereas a receive action moves a message between the two queues.
        Since the two actions operate on different messages, performing them in either order yields the same result.

        Therefore, in all cases, we can derive co-final transitions.
    \end{proof}

    \begin{lemma}[Backward transitions are independent]
        \label{lem:backward-independent}
        Given $\program$, suppose that $\Ini^{\program} \genBStar{} \Sys$, $\Sys \genB{\overline{\alpha_1}} \Sys_1$ and $\Sys \genB{\overline{\alpha_2}} \Sys_2$,
        and $\overline{\alpha_1} \neq \overline{\alpha_2}$.
        Then, $\overline{\alpha_1} \idpB \overline{\alpha_2}$.
    \end{lemma}
    \begin{proof}
        We consider the same cases as in the proof of the Square property.

        First, if both transitions are actions of the same process, then backward execution is uniquely determined by the history item to be consumed.
        This implies that the two transitions coincide, which contradicts the assumption $\overline{\alpha_1} \neq \overline{\alpha_2}$.

        Second, suppose that one transition corresponds to the creation of the other process.
        Since the child process can still perform a backward transition, its history cannot be empty.
        However, undoing a channel creation requires that the child process undo its actions, which yields a contradiction.

        Finally, consider the case where both actions operate on the same channel.
        If both actions are send actions, then by the uniqueness of keys, they must be actions of the same process, and thus the two transitions coincide.
        The same argument applies if both actions are receive actions.
        If one action is a send and the other is a receive action, then a message cannot simultaneously be located in both the current-queue part and the history-queue part.
        Therefore, the two actions must have different keys.

        Therefore, in all cases, we conclude that $\overline{\alpha_1} \idpB \overline{\alpha_2}$ holds.
    \end{proof}

    \begin{lemma}[Well-foundedness]
        \label{lem:well-foundedness}
        There is no infinite backward computation: $\Sys_0 \genB{\overline{\alpha_1}} \Sys_1 \genB{\overline{\alpha_2}} \cdots$.
    \end{lemma}
    \begin{proof}
        This is immediate from the fact that histories are finite and that each backward transition necessarily consumes exactly one history item.
    \end{proof}

    We show that the semantics satisfies causal-consistent reversibility.
    \ogm{We use $d:\Sys_0\genBStar{\rho}\Sys_n$ for a derivation.}

    Applying the axiomatic approach~\cite{10.1145/3648474},
    the key reversibility properties of Causal consistency, Causal safety, and Causal liveness hold for the transitions.

    \begin{theorem}[Causal consistency]
        \label{thm:cc-basic}
        Suppose that $\Sys$ is reachable, $d_1 : \Sys \genBStar{\rho_1} \Sys'$, and $d_2 : \Sys \genBStar{\rho_2} \Sys'$.
        Then $d_1 \cequivN d_2$.
    \end{theorem}
    \begin{proof}
        We adapt the axiomatic approach from~\cite{10.1145/3648474}.
    \end{proof}

    \ogm{\og{Causal consistency ensures that if} two derivations are coinitial and cofinal, then they \og{are} causally equivalent.}

    \ogm{Let $\sharp(\rho,\beta) = |\{i \mid \beta_i = \beta\}| - |\{i \mid \beta_i = \overline{\beta}\}|$,
        which counts uncancelled occurrences of $\beta$ in $\rho$.}
    \begin{theorem}[\og{Causal safety and liveness}]
        \label{thm:csl-basic}
        Suppose that $\Sys$ is reachable, $\Sys \genB{\alpha} \Sys' \genBStar{\rho} \Sys''$, and $\overline{\alpha}$ does not occur in $\rho$.
        \begin{trivlist}
            \item[\og{(Safety)}]
            \ogm{If $\og{\Sys'''} \genB{\alpha} \og{\Sys''}$ for some $\Sys'''$, then $\alpha \idpB \beta'$ for every $\beta'\in\rho$ such that $\sharp(\rho,\beta') > 0$.}
            \item[\og{(Liveness)}]
            \ogm{If $\alpha \idpB \beta'$ for every $\beta'\in\rho$ such that $\sharp(\rho,\beta') > 0$, then $\Sys''' \genB{\alpha} \Sys''$ for some $\Sys'''$.}
        \end{trivlist}
    \end{theorem}
    \begin{proof}
        \begin{trivlist}
            \item[\og{(Safety)}]
            \ogm{In the forward direction, the indices of actions associated with the same process identifier strictly increase.
            Hence, if the same action $\beta'$ were to appear twice in a derivation, there would necessarily be an occurrence of its reverse $\overline{\beta'}$ in between.
            Therefore, from the assumption $\overline{\alpha}$ does not occur in $\rho$, we can also conclude that $\alpha$ does not occur in $\rho$.
            As a result, the number of occurrences of the event corresponding to the transition $\Sys \genB{\alpha} \Sys'$ in the derivation $\Sys' \genBStar{\rho} \Sys''$ is zero,
            in the sense of event occurrence defined in~\cite{10.1145/3648474}.
            Moreover, for any label $\beta'$ of a transition that occurs with non-zero occurrence in $\rho$, we have $\sharp(\rho,\beta') > 0$.
            Therefore, by applying the axiomatic approach from~\cite{10.1145/3648474}, the theorem follows.}
            \item[\og{(Liveness)}]
            \og{The proof follows the same argument as in the Safety part, by adapting the axiomatic approach from~\cite{10.1145/3648474}.}
        \end{trivlist}
    \end{proof}

    \ogm{\og{Causal safety guarantees that a} transition can be undone only when its reverse is independent of all subsequent transitions that causally depend on it.
    \og{Thus,} actions are undone only in an order consistent with causality.
    \og{Causal liveness states that whenever} the reverse of a transition is independent of all subsequent transitions, it can eventually be undone.
    We can always return to the corresponding $\Ini^{\program}$ in \og{an} order that does not violate causal safety.}

    \section{Rollback-and-Replay Semantics}
    \label{sec:rr-semantics}
    
    This section defines \og{\emph{rollback-and-replay} semantics}, following~\cite{LamiLSCF24}, \og{for reversible debugging}.
    We regard a reachable system configuration obtained by the basic semantics as a \emph{breakpoint}, and define how debugging proceeds from that configuration.
    \sy{Starting from a breakpoint, \emph{rollback} moves the execution back to the point of a specified action, while \emph{replay} executes the specified action forward to the breakpoint,
    preserving the same send-receive correspondence and selections.}

    \og{First, we} introduce \og{replayable processes and replayable channel queues by extending those in Section~\ref{sec:revgo} with} sequences for replay.
    \sy{During backward execution, a replayable process retains the sequence of history items corresponding to undone actions,
    while a replayable channel queue retains the sequence of unsent messages.}
    From the must-happen-before relation on actions performed up to the breakpoint, we derive \og{communication dependency, which identifies independent messages that may be reordered}.
    \og{We then compose replayable processes and channel queues to define a replayable system that allows independent messages to be reordered.}
    \sy{Finally, rollback and replay are defined as controlled executions of the replayable system, guided by action dependency, which combines must-happen-before and communication dependency.
    Table~\ref{tab:notation} summarises the notation used in the semantics.}

    \begin{table}[H]
        \centering
        \caption{Notation used in the semantics}
        \label{tab:notation}
        \setlength{\tabcolsep}{2pt}
        \renewcommand{\arraystretch}{1.0}
        \begin{tabular}{@{}p{0.20\textwidth}|c|p{0.6\textwidth}@{}}
            \hline
            $\pcfg{\varp}{h}{P}$ & Def.~\ref{def:process} & Process ($h$: history). \\
            $\qcfg{c}{\sigma}{\mu}{\mu'}$ & Def.~\ref{def:queue} & Channel queue ($\mu$: current; $\mu'$: history). \\
            $\pcfg{\varp}{h}[h']{P}$ & Def.~\ref{def:replayable-process} & Replayable process ($h'$: replay). \\
            $\qcfg{c}{\sigma}[\mu]{\mu'}{\mu''}$ & Def.~\ref{def:replayable-queue} & Replayable channel queue ($\mu$: replay). \\
            $\proB{\beta}$, $\queBRev{\alpha}$, $\genB{\beta}$ & Figs.~\ref{fig:process_fwd},\ref{fig:process_bwd},\ref{fig:queue_fwd},\ref{fig:config} & Process/queue/system transitions. \\
            $\proL{\beta}$, $\queLRev{\alpha}$, $\genL{\beta}$ & Figs.~\ref{fig:r_process},\ref{fig:r_queue},\ref{fig:replayable_system} & Replayable process/queue/system transitions. \\
            $\rollback{\hat{\Sys}}{\overline{\rho}}$, $\replay{\hat{\Sys}}{\rho}$, $\genR$ & Fig.~\ref{fig:replayrollback} & Rollback/replay (request $\rho$) and transition. \\
            $\alpha\mhb\alpha'$ & Def.~\ref{def:must-happen-before} & Must-happen-before on actions. \\
            $k\cdep{\hat{\Pi}}k'$, $k\bowtie_{\hat{\Pi}}k'$ & Def.~\ref{def:communication-dependency} & Immediate/general communication dependency. \\
            $\mu\qeq{\hat{\Pi}}\mu'$, $\hat{\Sys}\syseq\hat{\Sys}'$ & Def.~\ref{def:replayable-system-equivalence} & Message-sequence/system equivalence. \\
            $\alpha\dep{\hat{\Sys}}\alpha'$ & Def.~\ref{def:action-dependency} & Action dependency for rollback/replay. \\[0.1em]
            \hline
        \end{tabular}
    \end{table}

    \subsection{Replayable Process}
    \label{subsec:replayable-config}
    
    \begin{definition}
        \label{def:replayable-process}
        A replayable process is either active or inactive.
        An active replayable process is $\pcfg{\varp}{h}[h']{P}$, where $\pcfg{\varp}{h}{P}$ is a process and $h'$ is a sequence of history items for replay.
        An inactive replayable process is $\ipcfg{\varp}{h'}$, where $h'$ is a sequence of history items for replay.
        A replayable process configuration, $\hat{\Pi}$, is a parallel composition of replayable processes.
    \end{definition}

    When process creation is executed backward, the created subprocesses are turned into inactive processes, storing a sequence of history items for replay.
    They are reactivated with the sequence of history items when the corresponding process creation is executed again.
    $\mathsf{ad_p}$ adds empty sequences to a process configuration:
    $\adp{\miniprod_{i=1}^n \pcfg{\varp[i]}{h_i}{P_i}} = \miniprod_{i=1}^n \pcfg{\varp[i]}{h_i}[\varepsilon]{P_i}$.

    Fig.~\ref{fig:r_process} defines transitions over replayable process configurations,
    written $\hat{\Pi} \proL{\alpha} \hat{\Pi}'$ and $\hat{\Pi} \proL{\overline{\alpha}} \hat{\Pi}'$.
    For non-creation actions, \brname{LAct} moves the last history item to the sequence for replay,
        and \frname{LAct} redoes the action guided by the sequence for replay.
        \brname{LNwc} turns the subprocesses into inactive processes and \frname{LNwc} reactivates the subprocesses.

    \begin{figure}[H]
        \input{fig6_replayable_process}
        \caption{Forward and backward semantics for replayable process configurations}
        \label{fig:r_process}
    \end{figure}

    \subsection{Replayable Queues}
    
    \begin{definition}
        \label{def:replayable-queue}
        A replayable channel queue is either active or inactive.
        An active replayable channel queue is $\qcfg{c}{\sigma}[\mu]{\mu'}{\mu''}$, where $\qcfg{c}{\sigma}{\mu'}{\mu''}$ is a channel queue and $\mu$ is
        a sequence of messages for replay.
        An inactive replayable channel queue is $\iqcfg{c}{\sigma}{\mu}$, where $\mu$ is a sequence of messages for replay.
        A replayable queue configuration, $\hat{Q}$, is a parallel composition of replayable channel queues.
    \end{definition}

    In the backward direction, when a send action is undone, the unsent message is moved to the sequence for replay.
    When channel creation is undone, the replayable channel queue becomes inactive and stores the message sequence for replay.
    In the forward direction, communications proceed according to the sequence for replay.
    $\mathsf{ad_q}$ adds empty sequences to the queue configuration of a breakpoint:
    $\adq{\miniprod_{i=1}^n \qcfg{c_i}{\sigma_i}{\mu_i'}{\mu_i''}} = \miniprod_{i=1}^n \qcfg{c_i}{\sigma_i}[\varepsilon]{\mu_i'}{\mu_i''}$.

    Fig.~\ref{fig:r_queue} defines transitions over replayable queue configurations, written $\hat{Q} \queLRev{\extA{\varp}{i}} \hat{Q}'$.
    \rrname{LQNwc} makes a replayable channel queue inactive in the backward direction and reactivates it in the forward direction.
    \rrname{LQSnd} moves an unsent message into the sequence for replay in the backward direction, and moves it back in the forward direction.
    \rrname{LQRcv} is similar to \rrname{QRcv} in Section~\ref{sec:revgo}.

    \begin{figure}[H]
        \centering
        
    \begin{tabular}{rll}
        \rrname{LQNwc} & $\hat{Q} \mid \iqcfg{c}{\sigma}{\mu}\ \queLRevLong{\nwcA{\varp}{i}{c}{\seq{\varp}}}\ \hat{Q} \mid \qcfg{c}{\sigma}[\mu]{\varepsilon}{\varepsilon}$ \\[0.2em]

        \rrname{LQSnd} & $\hat{Q} \mid \qcfg{c}{\sigma}[\mu\cdot(k,v)]{\mu'}{\mu''}\ \queLRevLong{\sndA{\varp}{i}{c}{k}{v}}\ \hat{Q} \mid \qcfg{c}{\sigma}[\mu]{(k,v)\cdot\mu'}{\mu''}$ \\[0.2em]

        \rrname{LQRcv} & $\hat{Q} \mid \qcfg{c}{\sigma}[\mu]{\mu'\cdot(k,v)}{\mu''}\ \queLRevLong{\rcvA{\varp}{i}{c}{k}{v}}\ \hat{Q} \mid \qcfg{c}{\sigma}[\mu]{\mu'}{(k,v) \cdot \mu''}$ \\[0.2em]
    \end{tabular}

        \caption{Reversible semantics for replayable queue configurations}
        \label{fig:r_queue}
    \end{figure}

    \subsection{Replayable System with Message Reordering}

    \label{subsec:swap}
    Composing replayable processes and replayable channel queues, we redo actions according to sequences for replay.
    \begin{definition}[Replayable system configuration]
        A replayable system configuration $\hat{\Sys}$ is a pair $(\hat{\Pi},\hat{Q})$ where $\hat{\Pi}$ is a replayable process configuration and $\hat{Q}$ is a replayable channel queue configuration.
    \end{definition}

    Given a system configuration $(\Pi,Q)$, we define the replayable system configuration obtained by adding empty sequences:
    $\ad{(\Pi,Q)} = (\adp{\Pi},\adq{Q})$.
    Given a breakpoint $\Tar$, we use $\ad{\Tar}$ as the starting point of debugging.

    To derive dependencies between communications, we collect actions from history items of replayable processes:
    \begin{center}
        $\act{\pcfg{\varp}{h}[h']{P}} = \idx{\varp,h \cdot h'}$ and $\act{\ipcfg{\varp}{h}} = \idx{\varp,h}$,
    \end{center}
    where $\idx{\varp,\varr[1],\cdots,\varr[n]} = \{\alpha_1 \cdots \alpha_n\}$ such that $\alpha_i$ is $\nwcA{\varp}{i}{c}{\varp[1],\cdots,\varp[n]}$ if $\varr[i] = \nwcH{c}{\varp[1],\cdots,\varp[n]}$,
    $\sndA{\varp}{i}{c}{k}{\_}$ if $\varr[i] = \sndH{c}{k}{e}{slc}$, $\rcvA{\varp}{i}{c}{k}{\_}$ if $\varr[i] = \rcvH{c}{k}{P}{slc}$, and $\itnA{\varp}{i}$ otherwise.
    We write $\act{\hat{\Pi}}$ for the actions associated with $\hat{\Pi}$.
    Note that $\act{\hat{\Pi}} = \act{\hat{\Pi}'}$ for
    $\hat{\Pi} \proL{\alpha} \hat{\Pi}'$ and $\hat{\Pi} \proL{\overline{\alpha}} \hat{\Pi}'$.

    From these actions, we derive dependencies between communications as relations on message keys.
    For an action $\alpha$, let $\key{\alpha}=\{k\}$ if $\alpha$ is of the form $\sndA{\_}{\_}{\_}{k}{\_}$ or $\rcvA{\_}{\_}{\_}{k}{\_}$, and $\key{\alpha}=\varnothing$ otherwise.
    We say that $\alpha$ is \emph{associated} with $k$ if $\key{\alpha}=\{k\}$.
    \begin{definition}[Communication dependency]
        \label{def:communication-dependency}
        Given $\hat{\Pi}$, $k$ immediately depends on $k'$ in $\hat{\Pi}$, written $k \cdep{\hat{\Pi}} k'$,
        if there exist actions $\alpha_1,\cdots,\alpha_n \in\act{\hat{\Pi}}$ $(n \geq 1)$ such that
        \begin{center}
            $\key{\alpha_1}=\{k\}, \quad \key{\alpha_n}=\{k'\}, \quad \text{and} \quad \alpha_1 \mhb \cdots \mhb \alpha_n.$
        \end{center}
        $k$ depends on $k'$ in $\hat{\Pi}$ if $k \cdep{\hat{\Pi}}^* k'$.
        We write $k \bowtie_{\hat{\Pi}} k'$ if $k \cdep{\hat{\Pi}}^* k'$ or $k' \cdep{\hat{\Pi}}^* k$.\\
        $k$ and $k'$ are independent in $\hat{\Pi}$ if $k \centernot\bowtie_{\hat{\Pi}} k'$.
    \end{definition}

    A chain of $\mhb$ steps between actions associated with $k$ and $k'$ induces a one-step dependency between the corresponding communications.
    Its reflexive transitive closure $\cdep{\hat{\Pi}}^*$ accounts for indirect dependencies.

    We use this independence to define when message sequences can be reordered.
    Let $\qeq{\hat{\Pi}}$ be the smallest equivalence relation on message sequences such that
    $\mu_1(k,v)(k',v')\mu_2 \qeq{\hat{\Pi}} \mu_1(k',v')(k,v)\mu_2$
    whenever $k \centernot\bowtie_{\hat{\Pi}} k'$.
    \begin{definition}[Replayable system equivalence]
        \label{def:replayable-system-equivalence}
        We define $\syseq$ as the smallest equivalence relation on replayable system configurations satisfying:
        \begin{center}
            $(\hat{\Pi},\hat{Q}\mid\qcfg{c}{\sigma}[\mu_1'']{\mu_1}{\mu_1'})\ \syseq\ (\hat{\Pi},\hat{Q}\mid\qcfg{c}{\sigma}[\mu_2'']{\mu_2}{\mu_2'})$
        \end{center}
        if $\mu_1\qeq{\hat{\Pi}}\mu_2$, $\mu_1'\qeq{\hat{\Pi}}\mu_2'$, and $\mu_1''\qeq{\hat{\Pi}}\mu_2''$.
    \end{definition}

    Fig.~\ref{fig:replayable_system} defines forward and backward transitions over replayable system configurations,
    written $\hat{\Sys} \genL{\alpha} \hat{\Sys}'$ and $\hat{\Sys} \genL{\overline{\alpha}} \hat{\Sys}'$.
    For a reachable breakpoint $\Tar$, we apply this semantics to $\hat{\Tar} = \ad{\Tar}$ as the initial state of debugging.
    Using $\syseq$, the semantics can reorder messages of independent communications.
    \frname{LItn}, \brname{LItn}, \frname{LExt}, and \brname{LExt} are the replayable counterparts of the system rules.
    \frname{LEq} and \brname{LEq} use $\syseq$ to reorder independent communications.

    \begin{figure}[H]
        \centering
        \input{fig8_replayable_system}
        \caption{Forward and backward semantics for replayable system configurations}
        \label{fig:replayable_system}
    \end{figure}

    \subsection{Rollback-and-Replay System}
    
    \label{subsec:rr}
    Starting from a breakpoint with empty sequences,
    we perform rollback for a specified past action and replay for a specified undone action
    by guiding the replayable system to execute only the actions dependent on the specified action.

    \begin{definition}
        Given $\hat{\Pi}$ and a message sequence $\mu=(k_1,v_1)\cdots (k_n,v_n)$\ogm{,}\\
        $k_i \dashrightarrow_{\hat{\Pi},\mu} k_j$ iff
        $i<j$; $k_i \bowtie_{\hat{\Pi}} k_j$; and for $i<l<j$, $k_{l} \centernot\bowtie_{\hat{\Pi}} k_i$ and $k_{l} \centernot\bowtie_{\hat{\Pi}} k_j$.
    \end{definition}

    To execute actions in a causal order, we define action dependency from the must-happen-before relation and local key dependency.
    \begin{definition}[Action dependency]
        \label{def:action-dependency}
        For $\hat{\Sys} = (\hat{\Pi},\hat{Q})$, let $\dep{\hat{\Sys}}$ be the relation on $\act{\hat{\Pi}}$ defined by $\alpha_1 \dep{\hat{\Sys}} \alpha_2$ iff one of the following conditions holds:
        \begin{itemize}[leftmargin=1.1em]
        \item $\alpha_1 \mhb \alpha_2$;
        \item $\alpha_1 = \sndA{\_}{\_}{c}{k_1}{\_}$, $\alpha_2 = \sndA{\_}{\_}{c}{k_2}{\_}$, and $k_1 \dashrightarrow_{\hat{\Pi},\mu} k_2$; and
        \item $\alpha_1 = \rcvA{\_}{\_}{c}{k_1}{\_}$, $\alpha_2 = \rcvA{\_}{\_}{c}{k_2}{\_}$, and $k_1 \dashrightarrow_{\hat{\Pi},\mu} k_2$,
        \end{itemize}
        where $\qcfg{c}{\sigma}[\mu_1]{\mu_2}{\mu_3}\in\hat{Q}$ and $\mu=\mu_1\cdot\mu_2\cdot\mu_3$.
    \end{definition}

    Action dependency captures the execution order between dependent actions.
    If $\alpha \dep{\hat{\Sys}} \alpha'$, then $\alpha$ is redone before $\alpha'$ in the forward direction,
    and $\alpha'$ is undone before $\alpha$ in the backward direction.

    Finally, we define the rollback-and-replay system, which performs rollback or replay for a specified action by executing only the actions that depend on it.
    We define done actions for replayable processes:
    \begin{center}
        $\done{\pcfg{\varp}{h}[h']{P}} = \idx{\varp,h}$ and $\done{\ipcfg{\varp}{h}} = \varnothing$.
    \end{center}
    We write $\done{\hat{\Pi}}$ for the done actions associated with $\hat{\Pi}$.
    We write $\undone{\hat{\Pi}}$ to denote $\act{\hat{\Pi}} \setminus \done{\hat{\Pi}}$.

    \begin{definition}[Rollback-and-replay system]
        Given a replayable system configuration $\hat{\Sys}$ and a sequence of actions $\rho$ called the {\em request}, $\rollback{\hat{\Sys}}{\overline{\rho}}$ is the {\em rollback system} for $\rho$ and $\replay{\hat{\Sys}}{\rho}$ is the {\em replay system} for $\rho$, where the transitions are defined in Fig.~\ref{fig:replayrollback}.
    \end{definition}

    \begin{figure}[H]
        \input{fig9_replay_rollback}
        \caption{Rollback-and-replay semantics}
        \label{fig:replayrollback}
    \end{figure}

    A rollback of a breakpoint $\Tar$ starts from $\rollback{\ad{\Tar}}{\overline{\alpha}}$, where $\alpha$ is the specified forward action.
    \brname{Req} adds dependent actions to the request, and \brname{Sat} undoes the requested action.
    When the request becomes empty, the rollback is complete.
    Replay can be performed in the same way by \frname{Sat} and \frname{Req}.
    By specifying an arbitrary action, the system can move between the initial state and the breakpoint.

    \begin{example}
        We show a rollback for the division-by-zero error from $\ad{\Tar}$,
        where $\Tar$ is from Example~\ref{ex:multi}.
        Let $\alpha = \sndA{\varp[0]}{1}{c}{k_0}{0}$ and $\alpha' = \rcvA{\varp[2]}{4}{c}{k_0}{0}$.

        \begin{align*}
            \rollback{\ad{\Tar}}{\overline{\alpha}} =& \rollback{(\hat{\Pi}_0,\qcfg{c}{\mathtt{int}}[\varepsilon]{(k_n,1)\cdots(k_2,1)}{(k_0,0)(k_1,1)})}{\overline{\alpha}}\\
            \genR& \rollback{(\hat{\Pi}_0,\qcfg{c}{\mathtt{int}}[\varepsilon]{(k_n,1)\cdots(k_2,1)}{(k_0,0)(k_1,1)})}{\overline{\alpha} \cdot \red{\overline{\alpha'}}}\\
            \genR& \rollback{(\hat{\Pi}_1,\qcfg{c}{\mathtt{int}}[\varepsilon]{(k_n,1)\cdots(k_2,1)\red{(k_0,0)}}{(k_1,1)})}{\overline{\alpha}}\\
            \genR& \rollback{(\hat{\Pi}_2,\qcfg{c}{\mathtt{int}}[\red{(k_0,0)}]{(k_n,1)\cdots(k_2,1)}{(k_1,1)})}{\varepsilon}\\
        \end{align*}

        \noindent
        The rollback first adds $\overline{\alpha'}$ to the request and executes it.
        By reordering independent messages, $\overline{\alpha}$ is executed without undoing the send action of the messages produced by the second sender, which are independent of the error.
    \end{example}

    \section{Properties of Rollback-and-Replay Semantics}
    \label{sec:rr-properties}
    
    We prove that the replayable system is causal-consistent and message swapping is sound with respect to reachability.
    In reversible debugging, always remaining within reachable configurations is important for correctly identifying bugs.
    We show that the rollback-and-replay system satisfies minimality since our channel communication with message swapping is fully asynchronous.

    \subsection{Causal-Consistent Reversibility}
    \label{subsec:4A-rev}
    
    For $\rho = \beta_1 \cdots \beta_n$, we write $\hat{\Sys}_0 \genLStar{\rho} \hat{\Sys}_n$ if $\hat{\Sys}_{i-1} \genL{\beta_i} \hat{\Sys}_i$ for $1 \leq i \leq n$.
    The independence relation is defined from \og{the action dependency} at the given breakpoint $\hat{\Tar}$.
    \begin{definition}[Independence relation over actions]
        Given $\hat{\Tar}$, we define the independence relation on actions, written $\idpL_{\hat{\Tar}}$, as follows:
        For $\beta_1,\beta_2\in\mathcal{A}\cup\overline{\mathcal{A}}$, let $\alpha_1=\und{\beta_1}$ and $\alpha_2=\und{\beta_2}$.
        Then $\beta_1 \idpL_{\hat{\Tar}} \beta_2$ iff $\alpha_1 \ndep{\hat{\Tar}} \alpha_2$ and $\alpha_2 \ndep{\hat{\Tar}} \alpha_1$.
    \end{definition}

    \og{
        We use $t: \hat{\Sys} \genL{\beta} \hat{\Sys}'$ for a transition and $d: \hat{\Sys}_0 \genLStar{\rho} \hat{\Sys}_n$ for a derivation.
        For $t: \hat{\Sys} \genL{\beta} \hat{\Sys}'$, we use $\overline{t}$ to denote $\hat{\Sys}' \genL{\overline{\beta}} \hat{\Sys}$.
        The independence determined at $\hat{\Tar}$ induces the following causal equivalence on derivations.}

    \begin{definition}[Causal equivalence]
        Given $\hat{\Tar}$, let $\cequivL_{\hat{\Tar}}$ be the smallest equivalence relation on derivations closed under concatenation and satisfying:
        \begin{trivlist}
            \item[(Swap)] if $\hat{\Tar} \genLStar{\rho} \hat{\Sys}$, $t_1: \hat{\Sys} \genL{\beta_1} \hat{\Sys}_1$, $t_2: \hat{\Sys} \genL{\beta_2} \hat{\Sys}_2$,
            $t'_2: \hat{\Sys}_1 \genL{\beta_2} \hat{\Sys}'$, $t'_1: \hat{\Sys}_2 \genL{\beta_1} \hat{\Sys}'$, and $\beta_1 \idpL_{\hat{\Tar}} \beta_2$,
            then $t_1 t'_2 \cequivL_{\hat{\Tar}} t_2 t'_1$.
            \item[(Cancellation)] for every transition $t$, $t\,\overline{t} \cequivL_{\hat{\Tar}} \varepsilon$ and $\overline{t}\,t \cequivL_{\hat{\Tar}} \varepsilon$.
        \end{trivlist}
    \end{definition}

    \noindent
    \ogm{To show the reversibility properties for the replayable semantics, we follow the axiomatic approach~\cite{10.1145/3648474} as done for the basic semantics in Section~2.}

    \begin{lemma}[Loop lemma]
        \label{lem:loop-rp}
        Suppose that $\Tar$ is reachable, $\ad{\Tar} \genLStar{\rho'} \hat{\Sys}$, and $\ad{\Tar} \genLStar{\rho''} \hat{\Sys}'$.
        Then $\hat{\Sys} \genL{\beta} \hat{\Sys}'$ iff $\hat{\Sys}' \genL{\overline{\beta}} \hat{\Sys}$.
    \end{lemma}
    \begin{proof}
        The proof is obtained in the same way as the loop lemma for system configurations.
    \end{proof}

    We apply the axiomatic approach~\cite{10.1145/3648474}.
    To derive Causal consistency, Causal safety, and Causal liveness for both forward and backward actions,
    we establish the Square property, FTI and BTI, and Well-foundedness for both forward and backward computations.

    \begin{lemma}[Square property]
        \label{lem:square-property-rp}
        Suppose that $\Tar$ is reachable, $\hat{\Tar} = \ad{\Tar} \genLStar{\rho} \hat{\Sys}$, $\hat{\Sys} \genL{\beta_1} \hat{\Sys}_1$, $\hat{\Sys} \genL{\beta_2} \hat{\Sys}_2$, and $\beta_1 \idpL_{\hat{\Tar}} \beta_2$.
        Then, there is $\hat{\Sys}'$ such that $\hat{\Sys}_1 \genL{\beta_2} \hat{\Sys}'$ and $\hat{\Sys}_2 \genL{\beta_1} \hat{\Sys}'$.
    \end{lemma}
    \begin{proof}
        The proof follows the same argument as that for the square property of the basic semantics.
        In the case where both actions are send operations on the same channel, their keys are distinct.
        Therefore, by reordering the independent messages in the replayable queues, we can derive co-final transitions.
        The same argument applies to the case where both actions are receive operations.
    \end{proof}

    \begin{lemma}[FTI and BTI]
        \label{lem:fti-bti-rp}
        Suppose that $\Tar$ is reachable, $\hat{\Tar} = \ad{\Tar} \genLStar{\rho} \hat{\Sys}$, $\hat{\Sys} \genL{\beta_1} \hat{\Sys}_1$ and $\hat{\Sys} \genL{\beta_2} \hat{\Sys}_2$, $\beta_1 \neq \beta_2$,
        and $\beta_1,\beta_2$ are both forward transitions or both backward transitions.
        Then, $\beta_1 \idpL_{\hat{\Tar}} \beta_2$.
    \end{lemma}
    \begin{proof}
        When both $\beta_1$ and $\beta_2$ are backward,
        for all cases other than send--send and receive--receive, the same argument as in the backward-transition independence (BTI)
        of the basic semantics applies.
        In the case where both actions are send operations on the same channel, both keys $k_1$ and $k_2$ can be moved to the contents by reordering the queue.
        Hence, the two keys are not ordered and we obtain $\beta_1 \idpL_{\hat{\Tar}} \beta_2$.
        The same argument applies to the case where both actions are receive operations on the same channel,
        since both $k_1$ and $k_2$ can be moved to the right of $\circ$.

        When both $\beta_1$ and $\beta_2$ are forward,
        for all cases other than send--send and receive--receive, the same argument as in the backward case applies.
        For the send--send case, the sequences for replay in replayable queues ensure that dependent send actions cannot be co-initial.
        Similarly, for the receive--receive case, the sequences for replay in replayable queues ensure that dependent receive actions cannot be co-initial.
        Therefore, we can derive $\beta_1 \idpL_{\hat{\Tar}} \beta_2$.
    \end{proof}

    \begin{lemma}[Well-foundedness]
        \label{lem:wf-rp}
        There is neither an infinite forward computation $\hat{\Sys}_0 \genL{\alpha_1} \hat{\Sys}_1 \genL{\alpha_2} \cdots$
        nor an infinite backward computation $\hat{\Sys}_0 \genL{\overline{\alpha_1}} \hat{\Sys}_1 \genL{\overline{\alpha_2}} \cdots$
        on the replayable system.
    \end{lemma}
    \begin{proof}
        This is immediate from the fact that histories and sequences for replay are finite.
    \end{proof}

    We derive the causal-consistent reversibility properties based on the axiomatic framework~\cite{10.1145/3648474}.

    \begin{theorem}[Causal consistency]
        \label{thm:cc-rp}
        Suppose that $\Tar$ is reachable, $\hat{\Tar} = \ad{\Tar} \genLStar{\rho'_1} \hat{\Sys}$,
        $d_1:\hat{\Sys} \genLStar{\rho_1} \hat{\Sys}'$, and $d_2:\hat{\Sys} \genLStar{\rho_2} \hat{\Sys}'$.
        Then $d_1 \cequivL_{\hat{\Tar}} d_2$.
    \end{theorem}
    \begin{proof}
        We adapt the axiomatic approach from~\cite{10.1145/3648474}.
    \end{proof}

    \og{The following properties ensure that an action can be undone or redone exactly when the corresponding transition respects action dependency.}

    \begin{theorem}[\og{Causal safety and liveness}]
        \label{thm:csl-rp}
        Suppose that $\Tar$ is reachable, $\hat{\Tar} = \ad{\Tar} \genLStar{\rho} \hat{\Sys} \genL{\beta} \hat{\Sys}' \genLStar{\rho'} \hat{\Sys}''$,
        and $\overline{\beta}$ does not occur in $\rho'$.
        \begin{trivlist}
            \item[\og{(Safety)}]
            \ogm{If $\hat{\Sys}''' \genL{\beta} \hat{\Sys}''$ for some $\hat{\Sys}'''$, then
            $\beta \idpL_{\hat{\Tar}} \beta'$ for every $\beta' \in \rho'$ such that $\sharp(\rho',\beta') > 0$.}
            \item[\og{(Liveness)}]
            \ogm{If $\beta \idpL_{\hat{\Tar}} \beta'$ for every $\beta' \in \rho'$ such that $\sharp(\rho',\beta') > 0$,
            then $\hat{\Sys}''' \genL{\beta} \hat{\Sys}''$ for some $\hat{\Sys}'''$.}
        \end{trivlist}
    \end{theorem}
    \begin{proof}
        \begin{trivlist}
            \item[\og{(Safety)}]
            \og{When $\beta$ is a forward action, the proof follows the same argument as in the Safety part of Theorem~\ref{thm:csl-basic}, by adapting the axiomatic approach from~\cite{10.1145/3648474}.}
            \ogm{When $\beta$ is a backward action, we can apply the axiomatic approach by symmetrically reversing the roles of the forward and backward directions.}
            \item[\og{(Liveness)}]
            \og{The proof follows the same argument as in the Safety part.}
        \end{trivlist}
    \end{proof}

    \subsection{Soundness}
    \label{subsec:4B-sound}
    
    \sy{
    Given $\hat{\Sys}=(\hat{\Pi},\hat{Q})$, define $\rmlog{\hat{\Sys}}=(\Pi,Q)$, where $\Pi$ and $Q$ result from the removal of replay history information:
    $\pcfg{\varp[i]}{h_i}[h'_i]{P_i}$ is projected to
    $\pcfg{\varp[i]}{h_i}{P_i}$ for process $p_i$, and
    $\qcfg{c_i}{\sigma_i}[\mu_i]{\mu'_i}{\mu''_i}$
    to $\qcfg{c_i}{\sigma_i}{\mu'_i}{\mu''_i}$ for channel $c_i$.}

    Reordering messages can result in a greater number of configurations due to the equivalence of queues.
    The following theorem demonstrates that reordering independent messages does not expand the set of reachable process configurations.

    \begin{theorem}[Soundness]
        \label{thm:soundness}
        \og{Given $\program$}, suppose that $\Tar$ is reachable \og{from $\program$}, $\rho \in (\mathcal{A}\cup\overline{\mathcal{A}})^*$, and $\ad{\Tar} \genLStar{\rho} \hat{\Sys}$.
        Then there exists $\hat{\Sys}'$ such that $\hat{\Sys}' \syseq \hat{\Sys}$ and $\rmlog{\hat{\Sys}'}$ is reachable \og{from $\program$}.
    \end{theorem}

    \begin{proofsketch}
        The proof constructs a configuration reachable from $\program$ corresponding to the result of rollback by removing the undone actions from a forward action sequence reaching the breakpoint.
        Since independent messages may be reordered in the replayable system, the resulting replayable configuration need only be equivalent to $\hat{\Sys}$.

        By applying the axiomatic approach with the roles of forward and backward actions reversed, the derivation from $\ad{\Tar}$ to $\hat{\Sys}$ can be normalised to
        \[
            \ad{\Tar}=\hat{\Sys}_0
            \genL{\overline{\alpha_1}}\hat{\Sys}_1
            \cdots
            \genL{\overline{\alpha_n}}\hat{\Sys}_n=\hat{\Sys}.
        \]
        Since $\Tar$ is reachable from $\program$, there is a forward action sequence $\lambda$ such that $\Ini^{\program}\genBStar{\lambda}\Tar$.
        This derivation can be lifted to the replayable system as $\hatIni\genLStar{\lambda}\ad{\Tar}$, where $\rmlog{\hatIni}=\Ini^{\program}$.

        We proceed by induction on the backward action sequence.
        For each $i$, we construct an action sequence $\lambda_i$, a system configuration $\Sys'_i$, and a replayable system configuration $\hat{\Sys}'_i$ such that
        $\lambda_i\alpha_i\cdots\alpha_1$ is a permutation of $\lambda$,
        $\Ini^{\program}\genBStar{\lambda_i}\Sys'_i$,
        $\hatIni\genLStar{\lambda_i}\hat{\Sys}'_i$,
        $\hat{\Sys}'_i\syseq\hat{\Sys}_i$, and
        $\rmlog{\hat{\Sys}'_i}=\Sys'_i$.
        Thus, $\lambda_i$ is the part of the original forward execution that remains after the first $i$ undone actions have been moved to the end and removed.
        The base case follows by taking $\lambda_0=\lambda$, $\Sys'_0=\Tar$, and $\hat{\Sys}'_0=\ad{\Tar}$.

        For the induction step, we permute $\lambda_i$ into $\lambda_{i+1}\alpha_{i+1}$ so that the action corresponding to the next backward step occurs last.
        Message reordering is unnecessary except when $\alpha_{i+1}$ is a send or receive action constrained by the order of another message on the same channel.
        In these cases, the relevant send and receive actions are reordered together.
        If further changes are required, the reordering is propagated to the necessary pairs of send and receive actions; all actions to which it is propagated are communication-dependent.
        The construction preserves action dependency and FIFO consistency, so $\lambda_{i+1}$ remains executable from $\Ini^{\program}$ in the basic semantics and satisfies the induction conditions for $i+1$.

        Taking $i=n$ yields $\Ini^{\program}\genBStar{\lambda_n}\Sys'_n$, $\hat{\Sys}'_n\syseq\hat{\Sys}$, and $\rmlog{\hat{\Sys}'_n}=\Sys'_n$.
        Hence, $\rmlog{\hat{\Sys}'_n}$ is reachable from $\program$.
        The details of the action-sequence construction and the preservation of action dependency and FIFO consistency are given in Appendix~\ref{sec:appendix-soundness}.
    \end{proofsketch}

    \sy{
    When $\Sys\syseq\Sys'$, the processes in $\Sys$ and $\Sys'$ are identical. The system configurations generated
    during rollback may not always be reachable when messages are reordered. However, the processes within these configurations are guaranteed to
    appear in some reachable configuration. Since $\genLStar{\rho}$
    covers all configurations generated by $\genBStar{\rho}$, the set
    of processes remains unchanged by the reordering. To reach
    an erroneous process configuration, a debugger may reorder
    messages, provided that they are independent.}

    \subsection{Minimality}
    \label{subsec:4C-minimality}
    
    \sy{By reordering messages, it is shown that
    rollback and replay can be implemented using the minimal number of steps with respect to the dependency relation.
    When $\rollback{\hat{\Sys}}{\rho}\genR^\ast\rollback{\hat{\Sys}'}{\varepsilon}$, let $\Bsat{\rollback{\hat{\Sys}}{\rho}\genR^\ast\rollback{\hat{\Sys}'}{\varepsilon}}$ be the number of applications of \brname{Sat}, and
    $\Fsat{\replay{\hat{\Sys}}{\rho}\genR^\ast\replay{\hat{\Sys}'}{\varepsilon}}$ be the number of applications of \frname{Sat}.
    }

    To roll back an action $\alpha$ from $\hat{\Sys}'$, $\rollback{\hat{\Sys}'}{\overline{\alpha}}$ gives the shortest length of transitions to reverse the action.
    It is also the case that $\replay{\hat{\Sys}}{\alpha}$ gives the shortest length of transitions for replaying $\alpha$ from $\hat{\Sys}$.

    \begin{theorem}[Minimality]
        \label{thm:minimality-replay-and-replay}
        Suppose that $\Tar$ is reachable, $\rho\in\mathcal{A}^\ast$, and $\ad{\Tar} \genLStar{\rho} \hat{\Sys}$.
        \begin{adjustwidth}{0.5em}{0pt}
        \begin{trivlist}
            \item[\ogm{(Rollback)}]
            \ogm{If $n_b = \Bsat{\rollback{\hat{\Sys}}{\overline{\alpha}} \genR^* \rollback{\hat{\Sys}'}{\varepsilon}}$, then
            there is no $d: \hat{\Sys} \genLStar{\rho'} \hat{\Sys}''$ such that $\alpha \in \undone{\hat{\Sys}''}$ and $|\rho'| < n_b$.}
            \item[\ogm{(Replay)}]
            \ogm{If $n_f = \Fsat{\replay{\hat{\Sys}}{\alpha} \genR^* \replay{\hat{\Sys}'}{\varepsilon}}$, then
            there is no $d: \hat{\Sys} \genLStar{\rho'} \hat{\Sys}''$ such that $\alpha \in \done{\hat{\Sys}''}$ and $|\rho'| < n_f$.}
        \end{trivlist}
        \end{adjustwidth}
    \end{theorem}
    \begin{proof}
        During replay, by \frname{Req}, any action $\alpha'$ added to the request satisfies $\alpha' \dep{\hat{\Tar}}^{+} \alpha$, where $\alpha$ is the initially specified action.
        By applying \og{the Safety part of Theorem~\ref{thm:csl-rp}} for forward actions, it is impossible to execute $\alpha$ without executing all such actions.
        Similarly, during rollback, by \brname{Req}, any action $\alpha'$ added to the request satisfies $\alpha \dep{\hat{\Tar}}^{+} \alpha'$, where $\alpha$ is the initially specified action.
        By applying \og{the Safety part of Theorem~\ref{thm:csl-rp}} for backward actions, it is impossible to execute $\alpha$ without executing all such actions.
    \end{proof}

    \subsection{Discussion}
    \label{subsec:discussion}
    \ogm{\cite{DBLP:journals/lmcs/MezzinaP21} studies reversible semantics for asynchronous multiparty session types with higher-order messages
    over a single global queue,
    while this paper presents that for a concrete programming language with multiple channel queues.}
    \og{Unlike~\cite{DBLP:journals/lmcs/MezzinaP21}, revGo assigns unique keys to messages, enabling messages in each channel queue to be asynchronously reordered during rollback according to key independence.}
    \ogm{However, proving the soundness of reordering requires extra effort.}
    For example, consider the following revGo program:
    \begin{center}
        $\nwcS{\og{w},y,z}{\mathtt{int}}{\red{\sndS{\og{w}}{1}[\sndS{y}{2}]} \| \blue{\sndS{z}{3};\rcvS{\og{w}}{n_1}} \| \green{\rcvS{y}{n_2};\sndS{z}{4}}},$
    \end{center}
    where $\red{\sndS{\og{w}}{1}[\sndS{y}{2}]}$ sends 1 to $\blue{\sndS{z}{3};\rcvS{\og{w}}{n_1}}$ via $\og{w}$ and 2 to $\green{\rcvS{y}{n_2};\sndS{z}{4}}$ via $y$.
    \og{Here, $\mathtt{newchan}(w,y,z:\mathtt{int})$ abbreviates three \textsf{int}-typed channel creations.} 
    Messages on $z$ are sent into the queue only as $\green{\langle 4 \rangle}\blue{\langle 3 \rangle}$.
    According to the approach of~\cite{DBLP:journals/lmcs/MezzinaP21}, when a rollback is requested for $\blue{\sndS{z}{3}}$ after all processes terminate, both $\blue{\sndS{z}{3};\rcvS{\og{w}}{n_1}}$ should be undone.
    This is because the communication between the first and second processes via channel $\og{w}$ is independent.
    The key for the message in $\og{w}$ is independent of all other keys of the messages in $y$ and $z$.
    The action $\blue{\rcvS{\og{w}}{n_1}}$ can be undone only after $\red{\sndS{\og{w}}{1}}$ is itself ready to be \og{undone. This} is because the first process must return to the point where $\og{w}$ communication is undone.
    Similarly, $\red{\sndS{y}{2}}$, $\green{\rcvS{y}{n_2}}$, and $\green{\sndS{z}{4}}$ must be undone before reaching $\blue{\sndS{z}{3}}$, since they depend on $\red{\sndS{\og{w}}{1}}$.
    The revGo approach does not require undoing $\red{\sndS{y}{2}}$; $\blue{\sndS{z}{3};\rcvS{\og{w}}{n_1}}$ is undone without retracing process locations, because message dependencies are tracked as keys. 
    To make this rollback causally consistent, messages in $z$ must be reordered from $\green{\langle 4 \rangle}\blue{\langle 3 \rangle}$ to $\blue{\langle 3 \rangle}\green{\langle 4 \rangle}$.
    The revGo semantics abstracts this reordering using equivalence, whereas \cite{DBLP:journals/lmcs/MezzinaP21} addresses it explicitly.

    \section{Conclusion}
    \label{sec:conclusion}
    
    We introduced revGo, a reversible semantics for a core fragment of Go with asynchronous channel-based communication
    assuming unbounded queue lengths.
    By enriching processes and channel queues with histories and uniquely keyed messages,
    we formalised causal independence that allows safe reordering of independent keyed messages without affecting reachability.
    Based on an axiomatic framework for reversible computation~\cite{10.1145/3648474}, we proved causal consistency, safety, and liveness,
    ensuring that rollback respects causal dependencies.
    We further defined a rollback-and-replay semantics that covers all behaviour and enables a minimal
    number of undoing and redoing actions only causally related to a target event.
    This provides a rigorous foundation of reversible debugging for channel-based concurrent systems.

    There are several directions for future work.
    First, we would like to develop a practical reversible debugger prototype for programs that use channel queues,
    following the design of CauDEr.
    We would also like to evaluate its overhead and scalability~\cite{rr,rept}.
    Second, we would like to extend our semantics with bounded queues to adapt Go semantics.
    The assumption of unbounded queues eases the theoretical treatment
    while bounded queues will introduce a further element of synchronisation, as a full queue acts like
    a semaphore for processes that want to send to it.
    This will require a finer causal-tracking semantics.
    Another direction could be to extend a subset of the language to cope with other characteristics of Go
    such as \og{structs} or synchronisation primitives~\cite{GoMutex} on shared variables.

    \begin{credits}
        \subsubsection{\ackname}
        We thank the anonymous reviewers of ICTAC for their helpful comments and suggestions.
        This work was supported by JSPS Invitation Fellowship for Research in Japan (S25016), EPSRC EP/T006544/2, EP/T014709/2, EP/Z533749/1, ARIA, and Horizon EU TaRDIS 101093006 (UKRI number 10066667).
    \end{credits}

    \bibliographystyle{splncs04}
    \bibliography{mybibliography}

    \appendix

    \section{Proofs for Soundness (\ref{subsec:4B-sound})}
    \label{sec:appendix-soundness}
    \subsection{Backward Normalization}
    First, we show that any replayable system configuration reachable from a breakpoint by a sequence of backward and forward transitions is reachable by backward transitions alone.
    \begin{lemma}\label{lem:normalization}
        Suppose that $\rho \in (\fAct \cup \bAct)^*$, $\Tar$ is reachable, and
        $\ad{\Tar} \genLStar{\rho} \hat{\Sys}$.
        Then there exists $\rho' \in \bAct^*$ such that $\ad{\Tar} \genLStar{\rho'} \hat{\Sys}$.
    \end{lemma}
    \begin{proof}
        We apply the axiomatic approach~\cite{10.1145/3648474} with the roles of the forward and backward directions reversed, treating $\ad{\Tar}$ as irreversible.
    \end{proof}

    We show that there is a replayable system configuration corresponding to $\Ini^{\program}$,
    from which the breakpoint is reachable on the replayable system.
    \begin{lemma}\label{lem:lifting}
        Given $\program$, suppose that \og{$\rho \in (\fAct \cup \bAct)^*$,} $\lambda \in \fAct^*$,
        \og{$\Tar$ is reachable from $\program$, $\ad{\Tar}=\hat{\Tar} \genLStar{\rho} \hat{\Sys}$,}
        \og{$\rmlog{\hat{\Sys}} = \Sys$,} and $\Ini^{\program} \genBStar{\lambda} \og{\Sys}$.
        Then there exists $\hatIni$ such that $\rmlog{\hatIni} = \Ini^{\program}$ and $\hatIni \genLStar{\lambda} \hat{\Sys}$.
    \end{lemma}
    \begin{proof}
        By the semantics of the replayable system, starting from $\hat{\Sys}$,
        we can execute the sequence $\lambda$ backward in reverse order and reach some $\hatIni$
        such that $\rmlog{\hatIni} = \Ini^{\program}$, without using \brname{LEq}.
    \end{proof}

    \subsection{FIFO-consistency and Compatibility}

    Consider the case where there is a $\genBStar{}$-derivation labelled by $\lambda_1\,\red{\alpha}\,\lambda_2$ and a $\genLStar{}\,$-derivation labelled by $\lambda_1\lambda_2\,\red{\alpha}$.
    If $\red{\alpha}$ is not a channel action, causal safety and liveness immediately ensure that there is also a $\genBStar{}$-derivation labelled by $\lambda_1\lambda_2\,\red{\alpha}$.
    However, if $\red{\alpha}$ is a channel action, this need not hold, since the reordered sequence \og{may violate the first-in, first-out order}.
    We write $s_{c,k}$ and $r_{c,k}$ for send and receive actions,
    respectively, on channel $c$ with key $k$.
    \begin{definition}\label{def:fifo}
        Given $\lambda=\alpha_1,\cdots,\alpha_n \in \fAct^*$, we say that $\lambda$ is FIFO-consistent if
        $i<j \Leftrightarrow i'<j'$ whenever $\alpha_i = s_{c,k}$, $\alpha_j = s_{c,k'}$, $\alpha_{i'} = r_{c,k}$, and $\alpha_{j'} = r_{c,k'}$.
    \end{definition}

    FIFO-consistency means that, on each channel, send actions and their corresponding receive actions occur in the same order.
    As we show next, every $\genBStar{}$-derivation is labelled by a FIFO-consistent sequence.
    \begin{lemma}\label{lem:fifo}
        Given $\program$, suppose that \og{$\lambda \in \fAct^*$ and} $\Ini^{\program} \genBStar{\lambda} \Sys$.
        Then $\lambda$ is FIFO-consistent.
    \end{lemma}
    \begin{proof}
        The derivation $\Ini^{\program} \genBStar{\lambda} \Sys$ implies that $\lambda$ is FIFO-consistent by the semantics of system configurations.
    \end{proof}

    Conversely, if there is a $\genLStar{}\,$-derivation labelled by a FIFO-consistent sequence, then there is a $\genBStar{}$-derivation labelled by the same sequence.
    \begin{lemma}\label{lem:fifo-reachability}
        Given $\program$, suppose that \og{$\rho \in \bAct^*$, $\lambda \in \fAct^*$,}
        $\Tar$ is reachable from $\program$, $\hat{\Tar} = \ad{\Tar} \genLStar{\rho} \hat{\Sys}$,
        $\rmlog{\hatIni} = \Ini^{\program}$, $\hatIni \genLStar{\lambda} \hat{\Sys}$,
        and $\lambda$ is FIFO-consistent.
        Then there are $\hat{\Sys}_1$ and $\Sys_1$ such that $\hat{\Sys} \syseq \hat{\Sys}_1$, $\rmlog{\hat{\Sys}_1} = \Sys_1$,
        and $\Ini^{\program} \genBStar{\lambda} \Sys_1$.
    \end{lemma}
    \begin{proof}
        Since $\lambda$ is FIFO-consistent, $\hat{\Sys}_1$ is reached by the forward replayable semantics
        without using \frname{LEq}.
    \end{proof}

    We need to find a FIFO-consistent permutation of the sequence.
    In addition, the permutation must not violate dependencies.
    We use ${\prec}$ for a transitive relation over actions.
    \begin{definition}\label{def:compatibility}
        \og{Given $\lambda =\alpha_1,\cdots,\alpha_n \in \fAct^*$ and a transitive relation ${\prec} \subseteq \fAct^2$,}
        we say that $\lambda$ is compatible with ${\prec}$ if $\alpha_i \prec \alpha_j$ implies $i<j$.
    \end{definition}

    Note that the definition is given for a transitive relation since it is used in the proof below.
    Given a $\genLStar{}\,$-derivation labelled by $\lambda_1\lambda_2\,\red{\alpha}$, we need to find a sequence
    $\lambda_3$ such that $\lambda_3$ is a permutation of $\lambda_1\lambda_2$ and $\lambda_3\,\red{\alpha}$ is FIFO-consistent
    and compatible with $\dep{\hat{\Tar}}^+$, where $\hat{\Tar}$ is a breakpoint in the replayable system.

    \subsection{Propagation of Reordering Constraints}
    To reorder communication actions while preserving the property of being compatible with ${\prec}$, it may also be necessary to reorder other communication actions.
    For example, consider the following case, in which reordering the key pair $(k,k')$ affects the ordering of another key pair $(l,l')$:
    \[
        \lambda = \lambda_0\,\red{s_{c,k}}\,\lambda_1\,\red{s'_{c_1,l}}\,\lambda_2\,\red{s''_{c_1,l'}}\,\lambda_3\,\red{s'''_{c,k'}}\,\lambda_4
    \]
    where $s_{c,k} \prec s'_{c_1,l}$ and $s''_{c_1,l'} \prec s'''_{c,k'}$.
    Then any permutation of $\lambda$ that reverses the order of $k$ and $k'$ and remains compatible with ${\prec}$ must also reverse the order of $l$ and $l'$.
    Moreover, reversing the order of $l$ and $l'$ may in turn require reversing the order of another pair of keys.
    The same reasoning applies to receive actions.
    We formalise this propagation of reordering requirements as follows.

    \begin{definition}\label{def:propagation}
        Given a transitive relation ${\prec} \subseteq \fAct^2$, we define $\prpgStar{{\prec}}$ as the
        smallest reflexive and transitive relation over pairs of keys such
        that $(k,k') \prpgStar{{\prec}} (l,l')$ whenever one of the following
        holds:

        \begin{itemize}
            \item
            $s_{c,k} \preceq s'_{c_1,l}$ and $s''_{c_1,l'} \preceq s'''_{c,k'}$;
            \item
            $s_{c,k} \preceq r'_{c_1,l}$ and $r''_{c_1,l'} \preceq s'''_{c,k'}$;
            \item
            $r_{c,k} \preceq s'_{c_1,l}$ and $s''_{c_1,l'} \preceq r'''_{c,k'}$;
            \item
            $r_{c,k} \preceq r'_{c_1,l}$ and $r''_{c_1,l'} \preceq r'''_{c,k'}$.
        \end{itemize}
        Here, for actions $\alpha$ and $\alpha'$, $\alpha \preceq \alpha'$ if and only if
        $\alpha \prec \alpha'$ or $\alpha = \alpha'$.
    \end{definition}

    The four conditions distinguish the types of actions involved in one
    step of propagation. The first is the send-to-send case illustrated
    above: reordering the send actions for $(k,k')$ requires reordering the
    send actions for $(l,l')$. The second propagates the requirement
    from the send actions for $(k,k')$ to the receive actions for
    $(l,l')$. The third propagates it from the receive actions for
    $(k,k')$ to the send actions for $(l,l')$, and the fourth propagates
    it between the receive actions for the two pairs. Reflexivity includes
    the original pair itself, while transitivity captures the further
    propagation of the reordering requirement through other pairs of keys.

    The following property is important for ensuring that the order of the keys to be reordered is consistent.
    \begin{lemma}\label{lem:propagation}
        Suppose that \og{$\lambda \in \fAct^*$} is FIFO-consistent and compatible with ${\prec}$, and that $k$ and $k'$ appear in $\lambda$.
        Then $(k,k')\prpgStar{{\prec}}(l,l')$ implies $(k,k')\nprpgStar{{\prec}}(l',l)$.
    \end{lemma}
    \begin{proof}
        We prove by induction on the derivation of $(k,k')\prpgStar{{\prec}}(l,l')$ that $(k,k')\prpgStar{{\prec}}(l,l')$ implies $(k,k')\nprpgStar{{\prec}}(l',l)$.

        \paragraph{Base case}
        Since $\lambda$ is FIFO-consistent and compatible with ${\prec}$, we have
        $(k,k')\nprpgStar{{\prec}}(k',k)$.

        \paragraph{Induction step}
        Suppose that $(k,k')\prpgStar{{\prec}}(l_1,l'_1)$ and that
        one of the four generating conditions in the definition of
        $\prpgStar{{\prec}}$ propagates $(l_1,l'_1)$ to $(l,l')$.
        By the induction hypothesis, $(k,k')\nprpgStar{{\prec}}(l'_1,l_1)$.
        Suppose, towards a contradiction, that
        $(k,k')\prpgStar{{\prec}}(l',l)$.
        By symmetry of each generating condition, reversing both pairs
        propagates $(l',l)$ to $(l'_1,l_1)$.
        Hence, by transitivity, $(k,k')\prpgStar{{\prec}}(l'_1,l_1)$,
        contradicting $(k,k')\nprpgStar{{\prec}}(l'_1,l_1)$.
        Therefore, $(k,k')\nprpgStar{{\prec}}(l',l)$.
    \end{proof}

    \og{By Lemma~\ref{lem:propagation}, the propagated reordering requirements do not conflict with one another.}
    
    \og{If another communication action on the same channel occurs between the two communication actions to be reordered, the reordering requirement also propagates to that action.}
    For example, consider the following sequence:
    \[
        \lambda
        =
        \lambda_0\,
        \red{s_{c,k}}\,
        \lambda_1\,
        \red{s''_{c,k''}}\,
        \lambda_2\,
        \red{s'_{c,k'}}\,
        \lambda_3.
    \]
    We aim to reorder $\lambda$ so that \og{the send action for $k$ is the last on channel $c$}.
    Thus, \og{$\red{s_{c,k}}$ must be reordered not only with $\red{s'_{c,k'}}$ but also with the intervening $\red{s''_{c,k''}}$}.
    To preserve compatibility with ${\prec}$, \og{a related reordering constraint arises for} a key pair $(l,l')$ satisfying $(k,k')\prpgStar{{\prec}}(l,l')$,
    as illustrated by the sequence:
    \[
        \lambda=\lambda'_0\,\red{s_{c_1,l}}\,\lambda'_1\,\red{s''_{c_1,l''}}\,\lambda'_2\,\red{s'_{c_1,l'}}\,\lambda'_3
    \]
    \og{In this case, it suffices to reorder the send actions so that $\red{s'_{c_1,l'}}$ precedes $\red{s_{c_1,l}}$.}
    \og{The action $\red{s_{c_1,l}}$ need not be the last on channel $c_1$.}

    To define an algorithm for handling these reordering constraints, we first introduce two functions.
    For $\lambda$ and keys $l,l'$ appearing in $\lambda$, we define
    $\mathsf{itm}(\lambda,l,l')$ as the set of keys $l''$ such that either
    $s_{c,l}, s_{c,l'}, s_{c,l''} \in \lambda$ and $s_{c,l''}$ occurs between $s_{c,l}$ and $s_{c,l'}$ in $\lambda$,
    or $r_{c,l}, r_{c,l'}, r_{c,l''} \in \lambda$ and $r_{c,l''}$ occurs between $r_{c,l}$ and $r_{c,l'}$ in $\lambda$.
    For $\lambda$ and keys $l_1,l_2$ appearing in $\lambda$, we define $\mathsf{ord}(\lambda,l_1,l_2)$ as follows:
    \[
        \mathsf{ord}(\lambda,l_1,l_2)
        =
        \{(s_{c,l_1},s'_{c,l_2})\}
        \cup
        \{(r_{c,l_1},r'_{c,l_2})
        \mid r_{c,l_1},r'_{c,l_2}\in\lambda
        \text{ are receive actions}\}.
    \]
    Then we define the algorithm $\textproc{Partition}$ \og{to extend a given transitive relation with these auxiliary ordering constraints} as follows:

    \begin{algorithm}[H]
        \caption{\textproc{Partition}}
        \label{alg:partition}
        \begin{algorithmic}[1]
            \Require \og{$\lambda\in\fAct^*$; distinct keys $k,k'$ appearing in $\lambda$; a transitive relation ${\prec}_0\subseteq\fAct^2$}
            \Ensure \og{a transitive relation ${\prec}\subseteq\fAct^2$ such that ${\prec}_0\subseteq{\prec}$}
            \State ${\prec} \gets \og{{\prec}_0}$
            \ForAll{$k''\in\mathsf{itm}(\lambda,k,k')$}
                \State ${\prec}\gets({\prec}\cup\og{\mathsf{ord}(\lambda,k'',k')})^+$
            \EndFor
            \While{there are keys $l,l',l''$ such that $(k,k')\prpgStar{{\prec}}(l,l')$, $l''\in\mathsf{itm}(\lambda,l,l')$, $(k,k')\nprpgStar{{\prec}}(l,l'')$, and $(k,k')\nprpgStar{{\prec}}(l'',l')$}
                \State \og{Nondeterministically choose such keys $l,l',l''$ .}
                \If{a key $l_1$ appears in $\lambda$, $\exists l'_1.\,(k,k')\prpgStar{{\prec}}(l_1,l'_1)$, and $s_{c_1,l_1}\prec s'_{c_2,l''}$}
                    \State ${\prec}\gets({\prec}\cup\mathsf{ord}(\lambda,l,l''))^+$
                \Else
                    \State ${\prec}\gets({\prec}\cup\mathsf{ord}(\lambda,l'',l'))^+$
                \EndIf
            \EndWhile
            \State \Return ${\prec}$
        \end{algorithmic}
    \end{algorithm}

    \og{The \textbf{for} loop adds constraints to the transitive relation so that every communication action
    for a key $k''$ occurring between those for $k$ and $k'$ is reordered before the communication action for $k$.}
    \og{The \textbf{while} loop considers each propagated pair $(l,l')$ such that $(k,k')\prpgStar{{\prec}}(l,l')$ and
    ensures that every unconstrained key $l''\in\mathsf{itm}(\lambda,l,l')$ is explicitly ordered on one side or the other.}
    \og{If a key $l_1$ appears in $\lambda$ and there exists a key $l'_1$ such that $(k,k')\prpgStar{{\prec}}(l_1,l'_1)$ and $s_{c_1,l_1}\prec s'_{c_2,l''}$, the loop imposes the order $(l,l'')$; otherwise, it imposes the order $(l'',l')$.}

    We write $\mathsf{valid}_{\prec}(k,k')$ if, for every $(l,l')$ such that $(k,k')\prpgStar{{\prec}}(l,l')$, neither $s_{c,l}\prec s'_{c,l'}$ nor $r_{c,l}\prec r'_{c,l'}$ holds.
    \og{$\mathsf{valid}_{\prec}(k,k')$ means that, for every pair $(l,l')$ propagated from $(k,k')$,
    ${\prec}$ does not require a send or receive action associated with $l$ to precede the corresponding action associated with $l'$.}
    The following lemma shows that $\textproc{Partition}$ terminates and guarantees transitivity, compatibility, validity,
    and the assignment of every key to one of the two sides.

    \begin{lemma}\label{lem:partition}
        Given $\program$, suppose that $\lambda \in \fAct^*$, \og{$\rho \in (\fAct \cup \bAct)^*$,}
        $\Tar$ is reachable, $\ad{\Tar}=\hat{\Tar}=(\hat{\Pi},\hat{Q})$,
        $\hat{\Tar}\genLStar{\rho}\hat{\Sys}$, $\rmlog{\hat{\Sys}}=\Sys$,
        $\Ini^{\program}\genBStar{\lambda}\Sys$, $k$ and $k'$ appear in $\lambda$, and
        $k\centernot\bowtie_{\hat{\Pi}}k'$, and
        $k\centernot\bowtie_{\hat{\Pi}}k''$ for every
        $k''\in\mathsf{itm}(\lambda,k,k')$.
        Then $\textproc{Partition}(\lambda,k,k',\og{\dep{\hat{\Tar}}^+})$ terminates and returns
        a transitive relation ${\prec} \supseteq \og{{\dep{\hat{\Tar}}^+}}$ satisfying:
        \begin{itemize}
            \item $\lambda$ is compatible with ${\prec}$;
            \item $\mathsf{valid}_{\prec}(k,k')$; and
            \item if $(k,k')\prpgStar{{\prec}}(l,l')$ and $l''\in\mathsf{itm}(\lambda,l,l')$, then either $(k,k')\prpgStar{{\prec}}(l,l'')$ or $(k,k')\prpgStar{{\prec}}(l'',l')$.
        \end{itemize}
    \end{lemma}
    \begin{proof}
        By Lemma~\ref{lem:fifo} and $\Ini^{\program}\genBStar{\lambda}\Sys$, $\lambda$ is FIFO-consistent.
        \og{Let ${\prec}_0=\dep{\hat{\Tar}}^+$ and ${\prec}_1$ be the relation obtained after the first \textbf{for} loop of $\textproc{Partition}$.}
        The derivation $\Ini^{\program}\genBStar{\lambda}\Sys$ implies that $\lambda$ is compatible with ${\prec}_0$.
        Since the pairs added in the first \textbf{for} loop follow the order of $\lambda$, $\lambda$ is compatible with ${\prec}_1$.
        \og{By the independence assumptions, there is no dependency path in either
        direction between $k$ and $k'$, or between $k$ and any
        $k''\in\mathsf{itm}(\lambda,k,k')$.}
        \og{Therefore, for every $(l,l')$ such that
        $(k,k')\prpgStar{{\prec}_1}(l,l')$, neither
        $s_{c,l}\mathrel{{\prec}_1}s'_{c,l'}$ nor
        $r_{c,l}\mathrel{{\prec}_1}r'_{c,l'}$ holds.}
        \og{Hence, $\mathsf{valid}_{{\prec}_1}(k,k')$ holds.}

        \og{For $i\geq1$,} let ${\prec}_i$ be the current relation and ${\prec}_{i+1}$ the relation obtained after one iteration of the \textbf{while} loop.
        By the definition of $\textproc{Partition}$, ${\prec}_i$ is transitive.
        Since $\lambda$ is compatible with ${\prec}_i$, and the pairs added to construct ${\prec}_{i+1}$ follow the order of $\lambda$, $\lambda$ is compatible with ${\prec}_{i+1}$.
        Suppose that $\mathsf{valid}_{{\prec}_i}(k,k')$ and, towards a contradiction,
        that $\mathsf{valid}_{{\prec}_{i+1}}(k,k')$ does not hold because there exist
        $(l_1,l_1')$ such that
        $(k,k')\prpgStar{{\prec}_{i+1}}(l_1,l_1')$ and
        $s_{c,l_1}\mathrel{{\prec}_{i+1}}s'_{c,l'_1}$.
        Let $l,l',l''$ be the keys selected in this iteration of the
        \textbf{while} loop. Thus,
        $(k,k')\prpgStar{{\prec}_i}(l,l')$,
        $l''\in\mathsf{itm}(\lambda,l,l')$,
        $(k,k')\nprpgStar{{\prec}_i}(l,l'')$, and
        $(k,k')\nprpgStar{{\prec}_i}(l'',l')$.
        Since $\mathsf{valid}_{{\prec}_i}(k,k')$ holds, the
        ${\prec}_{i+1}$-ordering from $s_{c,l_1}$ to $s'_{c,l'_1}$ must use a
        pair of send actions added in this iteration. If
        $\mathsf{ord}(\lambda,l,l'')$ is added, the added pair is the one from
        $s_{c,l}$ to $s'_{c,l''}$. The parts of the ordering before and after
        this pair use ${\prec}_i$; hence, by the definition of propagation,
        $(l_1,l'_1)\prpgStar{{\prec}_i}(l,l'')$. Together with the propagation
        of $(l_1,l'_1)$ from $(k,k')$, this contradicts
        $(k,k')\nprpgStar{{\prec}_i}(l,l'')$.
        If $\mathsf{ord}(\lambda,l'',l')$ is added, the same argument gives
        $(l_1,l'_1)\prpgStar{{\prec}_i}(l'',l')$, contradicting
        $(k,k')\nprpgStar{{\prec}_i}(l'',l')$.
        The case where
        $r_{c,l_1}\mathrel{{\prec}_{i+1}}r'_{c,l'_1}$ is analogous.
        Therefore, $\mathsf{valid}_{{\prec}_{i+1}}(k,k')$ holds.

        Since the set of keys occurring in $\lambda$ is finite and each iteration monotonically extends the current relation by at least one new pair of actions, the \textbf{while} loop terminates.
        When the loop terminates, every key $l''$ of a propagated pair $(l,l')$ has been assigned to one of its two sides, so either $(l,l'')$ or $(l'',l')$ is propagated from $(k,k')$.
    \end{proof}

    \og{By Lemma~\ref{lem:partition}, $\lambda$ remains compatible with the transitive relation ${\prec}$ returned by $\textproc{Partition}$.}
    \og{For every pair $(l,l')$ propagated from $(k,k')$, ${\prec}$ does not require the send or receive action for $l$ to precede the corresponding action for $l'$,
        and for each $l'' \in \mathsf{itm}(l,l')$, either $(l,l'')$ or $(l'',l')$ is propagated from $(k,k')$ by $\prpgStar{}_{\prec}$.}

    \subsection{Reordering of Action Sequences}

    We define an operation that reorders two communication actions in a sequence while preserving compatibility with ${\prec}$.
    \begin{definition}\label{def:swap}
        Given a transitive relation \og{${\prec} \subseteq \fAct^2$}, \og{$\lambda \in \fAct^*$}, and \og{communication actions} $\alpha,\alpha' \in \lambda$ such that
        $\alpha$ precedes $\alpha'$ in $\lambda$, $\lambda$ is compatible with ${\prec}$, and $\alpha \mathrel{\not\prec} \alpha'$, we define $\mathsf{swp}_{{\prec}}(\lambda,\alpha,\alpha')\og{\in \fAct^*}$ as follows.
        Let $\lambda = \lambda' \alpha \lambda''' \alpha' \lambda''$, and divide $\lambda'''$ into three subsequences:
        \begin{itemize}
            \item $\lambda'''_1$, containing exactly those actions $\alpha_1$ such that $\alpha \prec \alpha_1$;
            \item $\lambda'''_2$, containing exactly those actions $\alpha'_1$ such that $\alpha'_1 \prec \alpha'$; and
            \item $\lambda'''_3$, containing all the remaining actions.
        \end{itemize}
        Then
        \[
            \mathsf{swp}_{{\prec}}(\lambda,\alpha,\alpha')
            =
            \lambda' \lambda'''_3 \lambda'''_2 \alpha' \alpha \lambda'''_1 \lambda''.
        \]
    \end{definition}

    We state that a sequence compatible with ${\prec}$ remains compatible with ${\prec}$ after applying $\mathsf{swp}$.
    \begin{lemma}\label{lem:compatibility}
        Suppose that \og{$\lambda \in \fAct^*$} is compatible with ${\prec}$, $\alpha,\alpha' \in \lambda$, $\alpha$ precedes $\alpha'$ in $\lambda$, and $\alpha \not\prec \alpha'$.
        Then $\mathsf{swp}_{{\prec}}(\lambda,\alpha,\alpha')$ is compatible with ${\prec}$.
    \end{lemma}
    \begin{proof}
        Since $\alpha,\alpha' \in \lambda$, $\alpha$ precedes $\alpha'$ in $\lambda$, and $\alpha \not\prec \alpha'$, $\mathsf{swp}_{{\prec}}(\lambda,\alpha,\alpha')$ is defined.
        By the definition of $\mathsf{swp}$, $\mathsf{swp}_{{\prec}}(\lambda,\alpha,\alpha')$ is compatible with ${\prec}$ if $\lambda$ is compatible with ${\prec}$.
    \end{proof}

    The propagation relation identifies pairs of keys whose order must be reversed together.
    For the reordered sequence to be FIFO-consistent, the send actions and, when present, the receive actions for each such pair must be ordered in the same direction.
    We define this local condition next.
    \begin{definition}\label{def:ordering}
        Given $\lambda=\alpha_1\cdots\alpha_n\in\fAct^*$ and keys $l$ and $l'$ that appear in $\lambda$,
        we say that $\lambda$ is $(l,l')$-ordered if
        \begin{itemize}
            \item $\alpha_i=s_{c,l}$ and $\alpha_j=s'_{c,l'}$ imply $i<j$; and
            \item $\alpha_{i'}=r_{c,l}$ and $\alpha_{j'}=r'_{c,l'}$ imply $i'<j'$.
        \end{itemize}
    \end{definition}

    The following lemma shows that if this condition holds in one direction or the other for every pair of keys on the same channel, then $\lambda$ as a whole is FIFO-consistent.
    \begin{lemma}\label{lem:ordering}
        Suppose that \og{$\lambda \in \fAct^*$ and} for any distinct keys $l$ and $l'$ on the same channel that appear in $\lambda$,
        either $\lambda$ is $(l,l')$-ordered or $\lambda$ is $(l',l)$-ordered.
        Then $\lambda$ is FIFO-consistent.
    \end{lemma}
    \begin{proof}
        By definition, for any distinct keys $l$ and $l'$ on the same channel that appear in $\lambda$,
        if either $\lambda$ is $(l,l')$-ordered or $\lambda$ is $(l',l)$-ordered,
        then the order of their send actions agrees with the order of their receive actions.
        Therefore, $\lambda$ is FIFO-consistent.
    \end{proof}

    We state a property about the order of keys unaffected by $\mathsf{swp}$.
    \begin{lemma}\label{lem:swap}
        Suppose that \og{$\lambda \in \fAct^*$} is compatible with ${\prec}$, that $k$ and $k'$ appear in $\lambda$, and that
        for any $l,l',l''$ such that $l''\in\mathsf{itm}(\lambda,l,l')$ and
        $(k,k')\prpgStar{{\prec}}(l,l')$,
        either $(k,k')\prpgStar{{\prec}}(l,l'')$ or
        $(k,k')\prpgStar{{\prec}}(l'',l')$ holds.
        For any $(l_1,l'_1)$ such that $(k,k')\prpgStar{{\prec}}(l_1,l'_1)$:
        \begin{itemize}
            \item $\mathsf{swp}_{{\prec}}(\lambda,s_{c,l_1},s'_{c,l'_1})$ is $(l,l')$-ordered if
            $\lambda$ is $(l,l')$-ordered and $(k,k')\nprpgStar{{\prec}}(l,l')$.
            \item $\mathsf{swp}_{{\prec}}(\lambda,r_{c,l_1},r'_{c,l'_1})$ is $(l,l')$-ordered if
            $\lambda$ is $(l,l')$-ordered and $(k,k')\nprpgStar{{\prec}}(l,l')$.
        \end{itemize}
    \end{lemma}
    \begin{proof}
        For a contradiction, suppose that $\lambda$ is $(l,l')$-ordered and
        $(k,k')\nprpgStar{{\prec}}(l,l')$, but that
        $(k,k')\prpgStar{{\prec}}(l_1,l'_1)$ and
        $\mathsf{swp}_{{\prec}}(\lambda,s_{c,l_1},s'_{c,l'_1})$ is not $(l,l')$-ordered.
        By the assumption that, for any $l,l',l''$ such that
        $l''\in\mathsf{itm}(\lambda,l,l')$ and $(k,k')\prpgStar{{\prec}}(l,l')$,
        either $(k,k')\prpgStar{{\prec}}(l,l'')$ or
        $(k,k')\prpgStar{{\prec}}(l'',l')$ holds,
        changing the order of the send or receive actions of $l$ and $l'$ in
        $\mathsf{swp}_{{\prec}}(\lambda,s_{c,l_1},s'_{c,l'_1})$ would imply
        $(k,k')\prpgStar{{\prec}}(l,l')$, a contradiction.
        The argument for $\mathsf{swp}_{{\prec}}(\lambda,r_{c,l_1},r'_{c,l'_1})$ is analogous.
    \end{proof}

    To prove soundness by induction on the backward action sequence,
    we use the following lemma.
    \begin{lemma}\label{lem:reorder}
        Given $\program$, suppose that \og{$\lambda \in \fAct^*$,} \og{$\rho \in (\fAct \cup \bAct)^*$,} $\rmlog{\hat{\Sys}} = \Sys$, $\ad{\Tar} = \hat{\Tar} \genLStar{\rho} \hat{\Sys}$,
        $\Ini^{\program} \genBStar{\lambda} \Sys$, and $\hat{\Sys} \genL{\overline{\alpha}} \hat{\Sys}'$.
        Then there exist $\lambda',\Sys'_1$, and $\hat{\Sys}'_1$ such that
        $\lambda' \alpha$ is a permutation of $\lambda$, $\hat{\Sys}' \syseq \hat{\Sys}'_1$, $\Sys'_1 = \rmlog{\hat{\Sys}'_1}$,
        $\Ini^{\program} \genBStar{\lambda'} \Sys'_1$, and $\hatIni \genLStar{\lambda'} \hat{\Sys}'_1$.
    \end{lemma}
    \begin{proof}
        By Lemma~\ref{lem:lifting}, there exists $\hatIni$ such that $\rmlog{\hatIni}=\Ini^{\program}$.
        Letting $\lambda'$ be the sequence obtained from $\lambda$ by removing $\alpha$, we have $\hatIni\genLStar{\lambda'}\hat{\Sys}'$ and $\Ini^{\program}\genBStar{\lambda'}\rmlog{\hat{\Sys}'}$.
        The essential cases are the following:
        \begin{itemize}
            \item $\alpha=r_{c,k}$, and the last receive on $c$ in $\lambda$ is $r'_{c,k'}$ with $k\neq k'$;
            \item $\alpha=s_{c,k}$, and the last send on $c$ in $\lambda$ is $s'_{c,k'}$ with $k\neq k'$.
        \end{itemize}
        We treat the two essential cases uniformly, using $c$, $k$, and $k'$ as specified above.

        By Lemma~\ref{lem:fifo} and $\Ini^{\program} \genBStar{\lambda} \Sys$, $\lambda$ is FIFO-consistent.
        Let $\hat{\Tar}=(\hat{\Pi},\hat{Q})$.
        Since $\hat{\Sys}\genL{\overline{\alpha}}\hat{\Sys}'$,
        $k\centernot\bowtie_{\hat{\Pi}}k'$ and
        $k\centernot\bowtie_{\hat{\Pi}}k''$ for every
        $k''\in\mathsf{itm}(\lambda,k,k')$.
        \og{Let ${\prec}_0=\dep{\hat{\Tar}}^+$, and let
        ${\prec}=\textproc{Partition}(\lambda,k,k',{\prec}_0)$.}
        By Lemma~\ref{lem:partition}, $\lambda$ is compatible with ${\prec}$,
        $\mathsf{valid}_{\prec}(k,k')$ holds, and every key
        of a propagated pair is assigned to one of its two sides.
        By Lemma~\ref{lem:propagation},
        $(k,k')\prpgStar{{\prec}}(l,l')$ implies
        $(k,k')\nprpgStar{{\prec}}(l',l)$.

        The following algorithm, $\textproc{Reorder}$, \og{constructs a permutation $\lambda'$ of $\lambda$ by applying $\mathsf{swp}_{\prec}$, when necessary, to the send and receive action pairs corresponding to each key pair propagated from $(k,k')$ under ${\prec}$}.
        \begin{algorithm}[H]
            \caption{\textproc{Reorder}}
            \label{alg:reorder}
            \begin{algorithmic}[1]
                \Ensure \og{$\lambda'\in\fAct^*$ such that $\lambda'$ is a permutation of $\lambda$}
                \State $\lambda' \gets \lambda$
                \ForAll{$(l,l')$ such that $(k,k')\prpgStar{{\prec}}(l,l')$}
                    \If{there are send actions $s_{c,l},s'_{c,l'}\in\lambda'$ such that $s_{c,l}$ precedes $s'_{c,l'}$}
                        \State $\lambda' \gets
                        \mathsf{swp}_{{\prec}}(\lambda',s_{c,l},s'_{c,l'})$
                    \EndIf
                    \If{there are receive actions $r_{c,l},r'_{c,l'}\in\lambda'$ such that $r_{c,l}$ precedes $r'_{c,l'}$}
                        \State $\lambda' \gets
                        \mathsf{swp}_{{\prec}}(\lambda',r_{c,l},r'_{c,l'})$
                    \EndIf
                \EndFor
                \State \Return $\lambda'$
            \end{algorithmic}
        \end{algorithm}


        \og{Let $\lambda'$ be the output of \textproc{Reorder}.}
        We list the set of pairs $\{(l,l')\mid(k,k')\prpgStar{{\prec}}(l,l')\}$ as
        $(l_1,l'_1),\ldots,(l_n,l'_n)$ in the order in which they are processed by $\textproc{Reorder}$.
        Let $\lambda_0=\lambda$, and for $1\leq i\leq n$, let $\lambda'_i$ be the sequence after processing the send actions in iteration $i$, and let $\lambda_i$ be the sequence after processing the receive actions.
        By Lemma~\ref{lem:compatibility}, each of $\lambda_i$ and $\lambda'_i$ is compatible with ${\prec}$.
        We show by induction that $\lambda_i$ is $(l'_j,l_j)$-ordered for all $j\leq i$.

        \paragraph{Base case}
        For $i=0$, there is no $j$ such that $1\leq j\leq i$, so the claim holds vacuously.

        \paragraph{Induction step}
        Assume that $\lambda_i$ is $(l'_j,l_j)$-ordered for all $j\leq i$.
        Since $(k,k')\prpgStar{{\prec}}(l_j,l'_j)$ implies $(k,k')\nprpgStar{{\prec}}(l'_j,l_j)$ for all $j\leq i$, Lemma~\ref{lem:swap} implies that
        $\lambda_{i+1}$ is $(l'_j,l_j)$-ordered for all $j\leq i$.
        Moreover, in $\lambda'_{i+1}$, we have the order
        $s'_{c,l'_{i+1}},s_{c,l_{i+1}}$, while in $\lambda_{i+1}$, we have the order
        $r'_{c,l'_{i+1}},r_{c,l_{i+1}}$.
        Suppose, for a contradiction, that $s_{c,l_{i+1}}$ precedes
        $s'_{c,l'_{i+1}}$ in $\lambda_{i+1}$.
        Then in $\lambda'_{i+1}$ we would have
        $r_{c,l_{i+1}},s_{c,l_{i+1}},s'_{c,l'_{i+1}},r'_{c,l'_{i+1}}$,
        contradicting the fact that $\lambda'_{i+1}$ is compatible with
        ${\prec}$ and that $\dep{\hat{\Tar}}\subseteq{\prec}$.
        Hence, $\lambda_{i+1}$ is $(l'_{i+1},l_{i+1})$-ordered.

        Therefore, $\lambda_n$ is $(l'_j,l_j)$-ordered for all $j\leq n$.

        If $\lambda$ is $(l,l')$-ordered and $(k,k')\nprpgStar{{\prec}}(l,l')$,
        then Lemma~\ref{lem:swap} implies that $\lambda_n$ is $(l,l')$-ordered.
        Therefore, for any distinct keys $l$ and $l'$ on the same channel that appear in $\lambda_n$,
        either $\lambda_n$ is $(l,l')$-ordered or $\lambda_n$ is $(l',l)$-ordered.
        Lemma~\ref{lem:ordering} yields that $\lambda_n$ is FIFO-consistent.

        Since $\lambda$ contains only finitely many keys, the \textbf{for} loop terminates.
        Thus, we obtain $\lambda_n$ such that $\lambda_n$ is a permutation of $\lambda$, and $\lambda_n$ is FIFO-consistent and compatible with ${\prec}$.
        It follows that $\lambda_n$ is compatible with $\dep{\hat{\Tar}}$, since $\dep{\hat{\Tar}}\subseteq{\prec}$.
        Finally, let $\lambda'\alpha$ be the permutation of $\lambda_n$ obtained by moving $\alpha$ to the end.
        Since $\dep{\hat{\Tar}}\subseteq{\prec}$, $\lambda'$ is FIFO-consistent and compatible with $\dep{\hat{\Tar}}$.
        Applying the axiomatic approach~\cite{10.1145/3648474} to the permutation $\lambda'\alpha$ yields a replayable system configuration $\hat{\Sys}'_1$ such that
        $\hat{\Sys}'\syseq\hat{\Sys}'_1$ and $\hatIni\genLStar{\lambda'}\hat{\Sys}'_1$.
        By Lemma~\ref{lem:fifo-reachability}, there exists $\Sys'_1$ such that
        $\Sys'_1=\rmlog{\hat{\Sys}'_1}$ and $\Ini^{\program}\genBStar{\lambda'}\Sys'_1$.

    \end{proof}

    \subsection{Proof of Soundness}

    We show the soundness of the replayable system.
    In the induction step, we use Lemma~\ref{lem:reorder} to preserve the induction invariant after each backward transition.
    \setcounter{theorem}{\numexpr\getrefnumber{thm:soundness}-1\relax}
    \begin{theorem}[Soundness]
        \label{thm:soundness-app}
        Given $\program$, suppose that $\Tar$ is reachable from $\program$, $\rho \in (\fAct\cup\bAct)^*$, and $\ad{\Tar} \genLStar{\rho} \hat{\Sys}$.
        Then there exists $\hat{\Sys}'$ such that $\hat{\Sys}' \syseq \hat{\Sys}$ and $\rmlog{\hat{\Sys}'}$ is reachable from $\program$.
    \end{theorem}
    \begin{proof}
        By Lemma~\ref{lem:normalization}, there exists $\rho' \in \bAct^*$ such that $\ad{\Tar} \genLStar{\rho'} \hat{\Sys}$.
        Since $\Tar$ is reachable from $\program$, there is $\lambda$ such that $\Ini^{\program} \genBStar{\lambda} \Tar$.
        By Lemma~\ref{lem:lifting},
        there exists $\hatIni$ such that $\rmlog{\hatIni} = \Ini^{\program}$ and $\hatIni \genLStar{\lambda} \ad{\Tar}$.
        Let $\rho' = \overline{\alpha}_1 \cdots \overline{\alpha}_n$ and
        $\ad{\Tar} = \hat{\Sys}_0 \genL{\overline{\alpha_1}} \hat{\Sys}_1 \genL{\overline{\alpha_2}} \cdots \genL{\overline{\alpha_n}} \hat{\Sys}_n = \hat{\Sys}$.
        We show by induction that there are $\lambda_i$, $\hat{\Sys}'_i$, and $\Sys'_i$
        such that $\lambda_i \alpha_i\cdots\alpha_1$ is a permutation of $\lambda$,
        $\hat{\Sys}'_i \syseq \hat{\Sys}_i$, $\rmlog{\hat{\Sys}'_i} = \Sys'_i$,
        $\hatIni \genLStar{\lambda_i} \hat{\Sys}'_i$, and $\Ini^{\program} \genBStar{\lambda_i} \Sys'_i$.

        \paragraph{Base case}
        Let $\lambda_0 = \lambda$, $\hat{\Sys}'_0 = \ad{\Tar}$, and $\Sys'_0 = \Tar$.
        Then, $\lambda_0$ is a permutation of $\lambda$, $\hat{\Sys}'_0 \syseq \hat{\Sys}_0$, $\rmlog{\hat{\Sys}'_0} = \Sys'_0$,
        $\hatIni \genLStar{\lambda_0} \hat{\Sys}'_0$, and $\Ini^{\program} \genBStar{\lambda_0} \Sys'_0$.

        \paragraph{Induction step}
        Suppose that $\hat{\Sys}_i \genL{\overline{\alpha}_{i+1}} \hat{\Sys}_{i+1}$ and
        that there are $\lambda_i$, $\hat{\Sys}'_i$, and $\Sys'_i$ such that
        $\lambda_i \alpha_i\cdots\alpha_1$ is a permutation of $\lambda$,
        $\hat{\Sys}'_i \syseq \hat{\Sys}_i$, $\rmlog{\hat{\Sys}'_i} = \Sys'_i$,
        $\hatIni \genLStar{\lambda_i} \hat{\Sys}'_i$, and $\Ini^{\program} \genBStar{\lambda_i} \Sys'_i$.
        Then $\hat{\Sys}'_i \genL{\overline{\alpha}_{i+1}} \hat{\Sys}_{i+1}$ follows by \brname{LEq}.
        By Lemma~\ref{lem:reorder},
        there are $\lambda_{i+1}$, $\hat{\Sys}'_{i+1}$, and $\Sys'_{i+1}$ such that
        $\lambda_{i+1} \alpha_{i+1} \alpha_i \cdots \alpha_1$ is a permutation of $\lambda$,
        $\hat{\Sys}'_{i+1} \syseq \hat{\Sys}_{i+1}$, $\rmlog{\hat{\Sys}'_{i+1}} = \Sys'_{i+1}$,
        $\hatIni \genLStar{\lambda_{i+1}} \hat{\Sys}'_{i+1}$, and $\Ini^{\program} \genBStar{\lambda_{i+1}} \Sys'_{i+1}$.

        Therefore, there are $\lambda_n$, $\hat{\Sys}'_n$, and $\Sys'_n$ such that
        $\lambda_n \alpha_n\cdots\alpha_1$ is a permutation of $\lambda$,
        $\hat{\Sys}'_n \syseq \hat{\Sys}$, $\rmlog{\hat{\Sys}'_n} = \Sys'_n$,
        $\hatIni \genLStar{\lambda_n} \hat{\Sys}'_n$, and $\Ini^{\program} \genBStar{\lambda_n} \Sys'_n$.
        Then, $\rmlog{\hat{\Sys}'_n}$ is reachable from $\program$.
    \end{proof}

\end{document}